\documentclass[draftcls,onecolumn,12pt]{IEEEtran}
\usepackage{color}
\usepackage{verbatim}
\usepackage{amsfonts}
\usepackage{amssymb}
\usepackage{stfloats}
\usepackage{cite}
\usepackage{graphicx}
\usepackage{psfrag}
\usepackage{subfigure}
\usepackage{amsmath}
\usepackage{array}
\usepackage{epstopdf}
\usepackage{authblk}
\usepackage{graphicx} 
\usepackage{amsthm} 
\usepackage{lipsum}
\usepackage{verbatim} 
\usepackage{authblk}
\usepackage{mathtools}
\usepackage{cuted}
\usepackage[lined,boxed,ruled]{algorithm2e}
\usepackage{algpseudocode}
\usepackage{framed} 
\usepackage{subfigure}
\usepackage{soul}
\usepackage{bm}
\usepackage{setspace}

\newtheorem{theorem}{Theorem}

\newtheorem{definition}{Definition}

\newtheorem{proposition}{Proposition}

\newtheorem{example}{\bf Example}
\newtheorem{remark}{\bf Remark}

\def\qed{$\Box$}

\def\phi{\varphi}

\def\l{\left}
\def\r{\right}
\def\({\left(}
\def\){\right)}

\def\bv{{\mathbf{v}}}
\def\bw{{\mathbf{w}}}

\def\b0{{\mathbf{0}}}

\def\cF{\mathcal{F}}

\def\cK{\mathcal{K}}

\def\cN{\mathcal{N}}
\def\cO{\mathcal{O}}

\def\cS{\mathcal{S}}

\newtheorem{Lemma}{Lemma}

\definecolor{LatestRevision}{rgb}{0.32, 0.18, 0.5}

\title{Resource Allocation for Multiuser Edge Inference with Batching and Early Exiting (Extended Version)}

\author{Zhiyan~Liu, \emph{Graduate Student Member, IEEE}, Qiao~Lan, \emph{Graduate Student Member, IEEE}, and Kaibin~Huang, \emph{Fellow, IEEE} 

\thanks{The work of K. Huang described in this paper was substantially supported by a fellowship award from the Research Grants Council of the Hong Kong Special Administrative Region, China (Project No. HKU RFS2122-7S04). The work was also supported by Guang-dong Basic and Applied Basic Research Foundation under Grant 2019B1515130003, Hong Kong Research Grants Council under Grants 17208319, and Shenzhen Science and Technology Program under Grant JCYJ20200109141414409.  (Corresponding author: Kaibin Huang).

Z. Liu, Q. Lan and K. Huang are with Department of Electrical and Electronic Engineering at The University of Hong Kong, Hong Kong. Contact: K. Huang  (Email: huangkb@eee.hku.hk).}}

\makeatletter
\newcommand{\removelatexerror}{\let\@latex@error\@gobble}
\makeatother

\IEEEoverridecommandlockouts

\begin{document}

\maketitle
\vspace{-10mm}
\begin{abstract}
    The deployment of inference services at the network edge, called edge inference, offloads computation-intensive inference tasks from mobile devices to edge servers, thereby enhancing the former’s capabilities and battery lives. In a multiuser system,  the joint allocation of communication-and-computation (C\textsuperscript{2}) resources (i.e., scheduling and bandwidth allocation) is made challenging by adopting efficient inference techniques, batching and early exiting, and further complicated by the heterogeneity in users’ requirements on accuracy and latency. Batching groups multiple tasks into a single batch for parallel processing to reduce time-consuming memory access and thereby boosts the throughput (i.e., completed task per second). On the other hand, early exiting allows a task to exit from a deep-neural network without traversing the whole network, thereby supporting a tradeoff between accuracy and latency. In this work, we study optimal C\textsuperscript{2} resource allocation with batching and early exiting, which is an NP-complete integer programming problem. %By tackling the challenge, a set of efficient algorithms are designed under the criterion of maximum throughput. 
    A set of efficient algorithms are designed under the criterion of maximum throughput by tackling the challenge.
    First, consider the case with batching but without early exiting. The target problem is solved optimally using a proposed best-shelf-packing algorithm that nests a threshold-based scheme, which selects users with the best channels and meeting the computation-time constraints, in a sequential search for the maximum batch size. Next, consider the general case with batching and early exiting. A low-complexity sub-optimal algorithm for C\textsuperscript{2} resource allocation is developed by modifying the preceding algorithm to exploit early exiting for latency reduction. On the other hand, the optimal approach is developed based on nesting a depth-first tree-search with intelligent online pruning into a sequential search for the maximum batch size. The key idea is to derive pruning criteria based on the simple greedy solution for the target problem without a bandwidth constraint and apply the result to designing an intelligent online pruning scheme. Experimental results demonstrate that both optimal and sub-optimal C\textsuperscript{2} resource allocation algorithms can leverage integrated batching and early exiting to double the inference throughput compared with conventional schemes. 
\end{abstract}


\section{Introduction}
\label{sec: introduction}
Edge inference refers to offloading inference tasks from edge devices to edge servers colocated and connected to \emph{base stations} (BSs) for execution using large-scale \emph{artificial intelligence} (AI) models\cite{Letaief2022JSAC,Zhang2020CM,Niu2019Infocom,Deniz2020SPAWC,LQ2021arxiv,Shao2021arxiv,Zhou2020IoTJ,Chen2021IOTJ,Chen2020TWC,dualTSP}. Leveraging powerful computation resources at servers, the devices’ capabilities can be dramatically enhanced (e.g., computer vision capable of discerning hundreds of object classes) and battery lives lengthened. This makes edge inference an important platform for realizing a wide range of AI-powered applications in the \emph{sixth-generation} (6G) mobile networks ranging from industrial automation to extended reality to autonomous vehicles and robots. Batching and early exiting are two key techniques for efficient multiuser edge inference. On one hand, batching assembles the tasks offloaded by multiple users into a single batch for parallel execution to improve the throughput, referring to the number of executed tasks per unit time, by amortizing memory access time and enhancing utilization of computation resources\cite{nvidia2018,TensorRT}. On the other hand, when traversing a deep neural network, early exiting allows a task to exit at a flexible depth, termed an exit point. This supports variable accuracy and latency to accommodate users’ heterogenous \emph{Quality-of-Service} (QoS) requirements as well as more efficiently utilize computation resources \cite{pmlr-v70-bolukbasi17a,MSDNet}. The implementation of multiuser edge inference is confronted with a communication bottleneck due to limited bandwidth, channel hostility, and high-dimensionality of privacy-preserving feature maps uploaded by multiple users. To overcome the communication bottleneck and limitation of computation resources at edge servers, we present in this work an efficient integration approach for joint management of \emph{communication-and-computation} (C\textsuperscript{2}) resources with the objective of maximizing the throughput of edge inference with batching and early exiting. 

Edge inference is an emerging area but fast expanding in different directions including joint source-channel coding\cite{Zhang2020CM}, feature pruning\cite{Niu2019Infocom,Deniz2020SPAWC}, progressive transmission\cite{LQ2021arxiv}, and distributed data compression using the information-bottleneck approach\cite{Shao2021arxiv}. The techniques have demonstrated effectiveness in tackling the communication bottleneck. Another research focus in the area of multiuser edge inference is joint C\textsuperscript{2} resource management. It is worth mentioning that the topic has been extensively studied in the broader area of \emph{mobile edge computing} (MEC) (see, e.g., the survey in\cite{GX2020CM}). While abstract computation models are typically adopted in MEC research, multiuser edge inference can be considered as MEC particularized to the edge-cloud services adopting specific AI models such as deep neural networks\cite{Letaief2022JSAC}.  Some initial research has been conducted in this area\cite{Zhou2020IoTJ,Chen2021IOTJ}. In\cite{Zhou2020IoTJ}, inference accuracy and transmission cost are optimally controlled as the \emph{Markov decision process} (MDP) to balance their tradeoff. On the other hand, considering a multi-core CPU computation model, the model partitioning point for split inference and computation-resource allocation are jointly controlled to minimize the maximum end-to-end latency across a set of multiuser inference tasks\cite{Chen2021IOTJ}. In view of prior work, multiuser edge inference is still a largely uncharted area. In particular, there is a lack of a design framework for joint C\textsuperscript{2} resource management for multiuser edge inference, let alone considering advanced techniques such as batching and early exiting.

Batching and early exiting are two basic techniques in the area of deep learning. First, consider batching. Its throughput improvement arises from \emph{weight reuse}, namely reusing model weights loaded into a processor for multiple inference tasks\cite{TensorRT}. This reduces the frequency of time-consuming access to memory storing a large-scale AI model and amortizes the memory access time over batched tasks, thereby boosting the throughput. Nevertheless, too large a batch size can result in excessive queueing time for individual users, which leads to throughput reduction. This requires optimization of the tradeoff by developing intelligent batching schemes such as  adapting the batch size to the estimated task-arrival intensity \cite{AliSC20} and dynamic management of batches\cite{LazyBatch,E2bird}. Next, consider early exiting that realizes a tradeoff balance between accuracy and computation latency based on a customized neural network architecture, called a \emph{backbone network}\cite{pmlr-v70-bolukbasi17a,MSDNet}. This architecture is a conventional deep neural network added with multiple low-complexity intermediate classifiers as candidate exit points. A task requiring a lower accuracy traverses a smaller number of network layers before exiting, namely being diverted to an intermediate classifier for immediate inference, and vice versa. Recently, early exiting has been applied to edge inference\cite{Chen2020TWC,dualTSP }. Specifically, the local/server model partitioning for split inference and exit points can be jointly optimized to maximize the inference accuracy under a latency constraint\cite{Chen2020TWC}; layer skipping and early exiting are combined to facilitate edge inference with hard resource constraints\cite{dualTSP}.
{In this work, we combine batching and early exiting in multiuser edge inference to boost its throughput. The advantages are twofold. First, the \emph{early-feedback effect}, i.e., instantaneous downloading of inference results upon early exits to intended users so as to shorten their task latency, is further complemented by batched parallel processing. Second, early exits release computation resources to accelerate other ongoing tasks in that batch as reflected in that a progressively reducing batch size over blocks of model layers accelerates computation over the blocks, termed the \emph{effect of shrinking batch size}.} However, the complex edge-inference process renders joint C\textsuperscript{2} resource management challenging as discussed shortly.

We consider an edge-inference system where an edge server responds to requests from multiple users to access a common AI model for executing their inference tasks under heterogeneous QoS requirements. In this system, C\textsuperscript{2} resource allocation involves scheduling a subset of users to grant their requests  and allocate bandwidth for uploading features. The optimization of control policies under the criterion of maximum system throughput is made challenging by three facts. First, batching and early exiting are interwound in that the latter results in a shrinking batch size over sequential layer blocks, causing heterogeneous computation time for different tasks. Second, besides communication coupling due to spectrum sharing, the scheduled users’ task executions are also inter-dependent in an intricate way. Specifically, the users’ accuracy requirements are translated into different numbers of layer blocks to traverse and the computation time of each block depends on the random number of tasks passing through it. Third, due to the constraints on end-to-end latency, communication and computation are coupled such that channel states and users’ QoS requirements need be jointly considered in batching/scheduling and bandwidth allocation. Mathematically, the optimization of C\textsuperscript{2} resource allocation involves an integer programming problem that is NP-complete.

In this work, we tackle the above challenges to design a framework of C\textsuperscript{2} resource allocation for efficient multiuser edge inference with batching and early exiting. {The design approach is characterized by the sequential treatment of two cases with increasing complexity, namely one using only batching and the other both batching and early exiting. This allows the algorithm and insight from solving the problem of optimal resource management in the former to be used as a stepping stone for efficiently computing the policies in the latter.} The key contributions and findings are summarized as follows.
\begin{itemize}
    \item \emph{Batched edge inference:} To begin with, we consider a simple case with batching but without early exiting, i.e., the full-network inference case. The problem of optimal C\textsuperscript{2} resource allocation is solved by transformation to allow the application of a best-shelf-packing algorithm. Specifically, consider the problem conditioned on a batch size. The optimal scheduling scheme is shown to select the users with the best channels from those whose latency budgets suffice for the expected batch computation time. Then the efficient best-shelf-packing algorithm nests the scheduling scheme in a sequential search of the maximum feasible batch size, yielding the maximum throughput.
    \item \emph{Edge inference with batching and early exiting:} 	First, a low-complexity sub-optimal algorithm for C\textsuperscript{2} resource allocation is proposed via omitting the batch size variation during inference and modifying its counterpart in the preceding case without early exiting. The resultant algorithm exploits the effect of early feedback but not the effect of shrinking batch size. Next, the optimal problem is solved by constructing a search tree for a given batch size and designing an intelligent depth-first search algorithm. To this end, a relaxed version of the problem without the bandwidth constraint is solved. The resultant greedy scheduling scheme greedily selects users with the lowest accuracy requirements. Building on the scheme, an intelligent algorithm for online tree pruning is proposed for the original problem to dramatically reduce the search complexity. The overall throughput maximization algorithm nests the efficient tree-search in a sequential search over the optimal batch size. 
    \item \emph{Performance: } Experimental results demonstrate more than 100\% throughput improvement achieved by the proposed design as opposed to the cases without dynamic batching or early exiting. Moreover, the proposed efficient tree-search algorithm achieves more than 10-fold complexity reduction with respect to searching without pruning. 
\end{itemize}

The rest of this paper is organized as follows. Section~\ref{sec: models_metrics} introduces the models and metrics. Section~\ref{sec: problem_formulation} formulates the optimization problems for multiuser edge inference. In Section~\ref{sec: MEAIwoEE}, we present the optimal algorithm for the full-network inference case, while in Section~\ref{sec: MEAIwEE}, we address the early exiting case and propose two algorithms. Experimental results are given in Section~\ref{sec: experiments}, followed by the conclusions in Section~\ref{sec: conclusions}. 

\section{Models and Metrics}
\label{sec: models_metrics}

Consider a single-cell system where an edge server provides the edge-inference service to multiple edge devices, as illustrated in Fig.~\ref{fig_system}.  These on-device applications are assumed to share a common trained prediction  model at the server, e.g., a large-scale classifier capable of discerning hundreds of object classes. Upon the arrival of  an inference  task, a device sends to the server a service  request with QoS requirements on   latency and accuracy. The request  is either granted  or denied due to the lack of necessary communication or computation resources. In the former case, based on  split inference, the scheduled device uploads a feature vector extracted from a multimedia file using  a local model, and then downloads the inference result when ready. The following protocol for edge inference is assumed. As shown in Fig.~\ref{fig_timeslot}, time is divided into \emph{epochs}, each of which is further divided into a communication and a computation slot. The slots have equal duration denoted as $T$. Each epoch has its computation slot overlapping  with the communication slot of the following  epoch to maximize the efficiency of resource utilization (see Fig.~\ref{fig_timeslot}). This setting and the alternative are discussed in the following remark.
\begin{remark}[\emph{Arrangements of communication and computation slots}] {\emph{The current setting of equal durations for communication and computation slots targets the scenario where the latency of communication and computation is comparable, necessitating their integrated design. On the other hand, the current design framework can also accommodate unequal durations for the two slot types by straightforward modifications. Specifically, given the durations of communication and computation slots, the current analysis and algorithm designs can be extended by modifying  each task's bandwidth requirement and latency budget accordingly, which does not change the form of the resource management problem and its solution approach.}}
\end{remark}

Consider an arbitrary epoch. At its beginning, the server receives service requests from $K$ active devices with their indices denoted  as  $\cK = \{1,2,\ldots,K\}$. The scheduling criterion  of the server is towards throughput maximization, namely maximizing the number of completed tasks in this epoch under their QoS requirements. At the end of the communication   slot used for uploading by scheduled devices, the edge server assembles   the received feature vectors for different tasks into a single batch and starts inference on the batch. As soon as it is computed, the inference result (i.e., a label) is fed back to the intended device. Models and metrics are described as follows. 

\begin{figure*}[!t]
  \centering
  \includegraphics[width=0.85\textwidth]{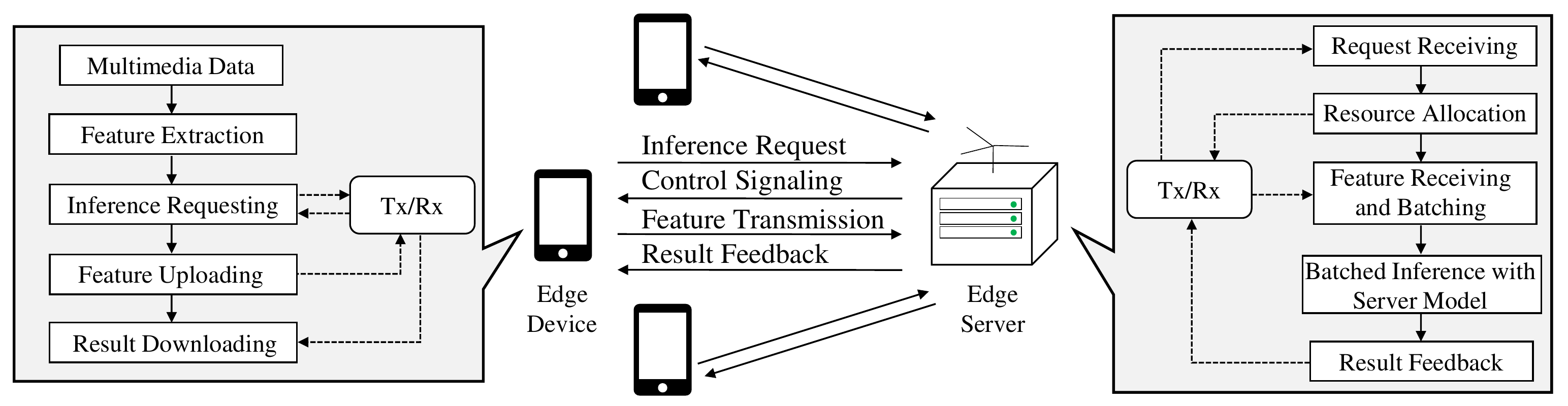}
  \caption{The edge inference system and operations. }\label{fig_system}
  \centering
\end{figure*}
\begin{figure*}[!t]
  \centering
  \includegraphics[width=0.76\textwidth]{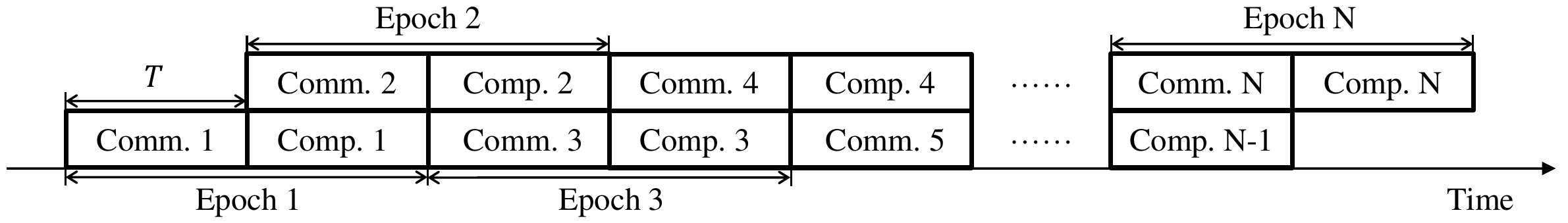}
  \caption{Timeline for edge inference. The n-th communication slot is denoted as ``Comm. n'' while the n-th computation slot is denoted as ``Comp. n''. }\label{fig_timeslot}
  \centering
\end{figure*}

\subsection{Communication Model}
As shown in Fig.~\ref{fig_system}, the communication between devices and the server involves  feature uploading, result downloading, and service requesting. The last two have negligible data sizes as opposed to feature uploading, which is hence the focus of the communication model. The system adopts \emph{orthogonal frequency-division multiple access} (OFDMA) for allocating a broadband spectrum to scheduled devices. The number of sub-carriers is assumed to be sufficiently large (e.g., several thousands  in 5G) such that bandwidth partitioning can be approximated as being continuous. Then let  $B$ denote the total bandwidth and $\rho_k$ with $0\leq\rho_k\leq1$ the fraction  of bandwidth allocated to device  $k$. Consider an arbitrary scheduled device, say the $k$-th one. The associated  channel is assumed frequency non-selective and thus the channel gain can be represented by a scalar  $h_k$, which remains constant within one epoch and is perfectly estimated by the server. The transmission power is represented  by $p_k$, which is also known by the server. Then the uplink communication rate,  $r_k$, can be written as 
\begin{equation}
    r_k = \rho_k B \log_2\l(1+\frac{p_k h_k^2}{N_0}\r),
\end{equation}
where $N_0$ is the white Gaussian channel noise power.  Let $\ell_k$ denote the size of the feature vector in bits. Since the feature vector must be uploaded within one communication slot (mathematically,  $r_k T \geq \ell_k$),  this sets the minimum fraction of allocated bandwidth, denoted $\rho_{\min, k}$, as
\begin{equation}
    \rho_{\min, k} \triangleq \frac{\ell_k}{{T B \log_2(1+p_k h_k^2/N_0)}}, 
\end{equation}
such that $\rho_k \geq \rho_{\min, k}$, if device $k$ is scheduled.

\subsection{Inference Model}

The edge server operates  a trained CNN model, which comprises multiple convolutional and fully-connected layers, to predict a label for each feature vector. The  batching and early exiting techniques are employed for the purposes of high efficiency and accommodating heterogeneous QoS requirements, respectively. Their models are described as follows. 

\subsubsection{Batching Model}
The server's GPUs operate in the batch mode to leverage their parallel computing capability and amortize the memory access time. At the start of a  computation slot, the edge server assembles  all received feature vectors into a single batch, which is a higher-dimensional tensor\cite{TensorRT}. The batch is then fed as a whole into the model for parallel inference. The batch size, denoted as $n_{\mathsf{b}}$, is a variable with support $\{1,\ldots,K\}$.  As found in existing experiments, there exists a deterministic relationship between the batch execution time, denoted as $t_\mathsf{cp}$, and the batch size for a given  specific  hardware platform (see e.g., \cite{AliSC20,LazyBatch,Inoue2021}).  Such a relation is represented by a mapping $f$: $t_{\sf cp}=f(n_{\mathsf{b}})$, which is a monotonically increasing function. For example, for an NVIDIA Tesla P4 GPU, the function is observed to be approximately \emph{linear} for different CNN architectures (i.e., GoogLeNet, ResNet-50, VGGNet-19) \cite{nvidia2018}, usually with a large constant term and a moderate slope. This results in that the throughput as measured by the number of executed tasks per second increases  as the batch size grows.

\subsubsection{Early Exiting Model}
Following the common design in the literature as illustrated in Fig.~\ref{fig_diagram_earlyexit}, the CNN model at the edge server is divided into $D$ sequential blocks of layers, each followed by a classifier that can predict the label using the extracted features\cite{pmlr-v70-bolukbasi17a,MSDNet}. A classifier by itself is a shallow  neural network comprising  a small number of layers as compared with  that in a block and  hence has much lower complexity. Based on this architecture, the server executes only a small number of  blocks before exiting via an intermediate  classifier for  a task requiring  relatively low accuracy/latency, and vice versa. Thereby the model supports a tradeoff between inference accuracy and latency. Let $1\leq d_k \leq D$ denote the exit point of task $k$, which is defined as the  number of executed blocks before exiting. For  ease of notation, we divide the model execution time into block-wise execution time. To this end,  let $t_d$ with $d=1,\ldots, D$ denote the execution time of the $d$-th block and $n_d=1,\ldots,D$ denote the remaining batch size upon entering the $d$-th block with $n_1=n_{\mathsf{b}}$. Then the sequence $n_1, n_2, \cdots, n_D$ is non-increasing. Let $f_d$ represent the mapping from the batch size to the block-wise execution time: $t_d=f_d(n_d)$. Since $n_d$ represents the computation load, the function is a monotonically increasing function. 

\subsection{QoS Metrics}
\subsubsection{End-to-end Latency}

The end-to-end latency of an inference task is measured from the moment the device sends the task request to the server until the downloading  of the inference result from the server. Consider an arbitrary task say, task $k$. It is considered successful only if the end-to-end latency does not exceed its constraint $\tau_k$.  The downloading latency is considered negligible due to the small sizes of inference results (i.e., labels). Consequently,  the end-to-end latency for task $k$ consists of three parts: 1) the waiting time from the task arrival to the beginning of an epoch, denoted by $t_{\mathsf{wt},k}$, 2) the time for feature uploading fixed as $T$ (see Fig.~\ref{fig_timeslot}), $t_{\mathsf{ul}, k}$, and 3) the computation time at the server, $t_{\mathsf{cp},k}$. At the beginning of a particular epoch, $t_{\mathsf{wt},k}$ is given. In the full-network inference case, the computation time $t_{\mathsf{cp},k}$ is the same for all scheduled tasks in the epoch, i.e., $t_{{\sf cp},k}=f(n_{\mathsf{b}})$. In the early exiting case, the inference result for task $k$ is ready as soon as the first $d_k$ blocks are computed, and thus $t_{\mathsf{cp},k}=\sum_{d=1}^{d_k}t_d = \sum_{d=1}^{d_k} f_d(n_d)$. Note that the computation time $t_{\mathsf{cp},k}$ relies on the exit point $d_k$, and thus varies over different tasks even in the same epoch. Tasks that early exit  enjoy lower end-to-end latency since the model is partially executed. Furthermore, as tasks early exit from the batch, the remaining tasks are also accelerated due to the release of computation  resource by exited tasks. 

\subsubsection{Inference Accuracy}
\begin{figure*}[!t]
  \centering
  \includegraphics[width=0.75\textwidth]{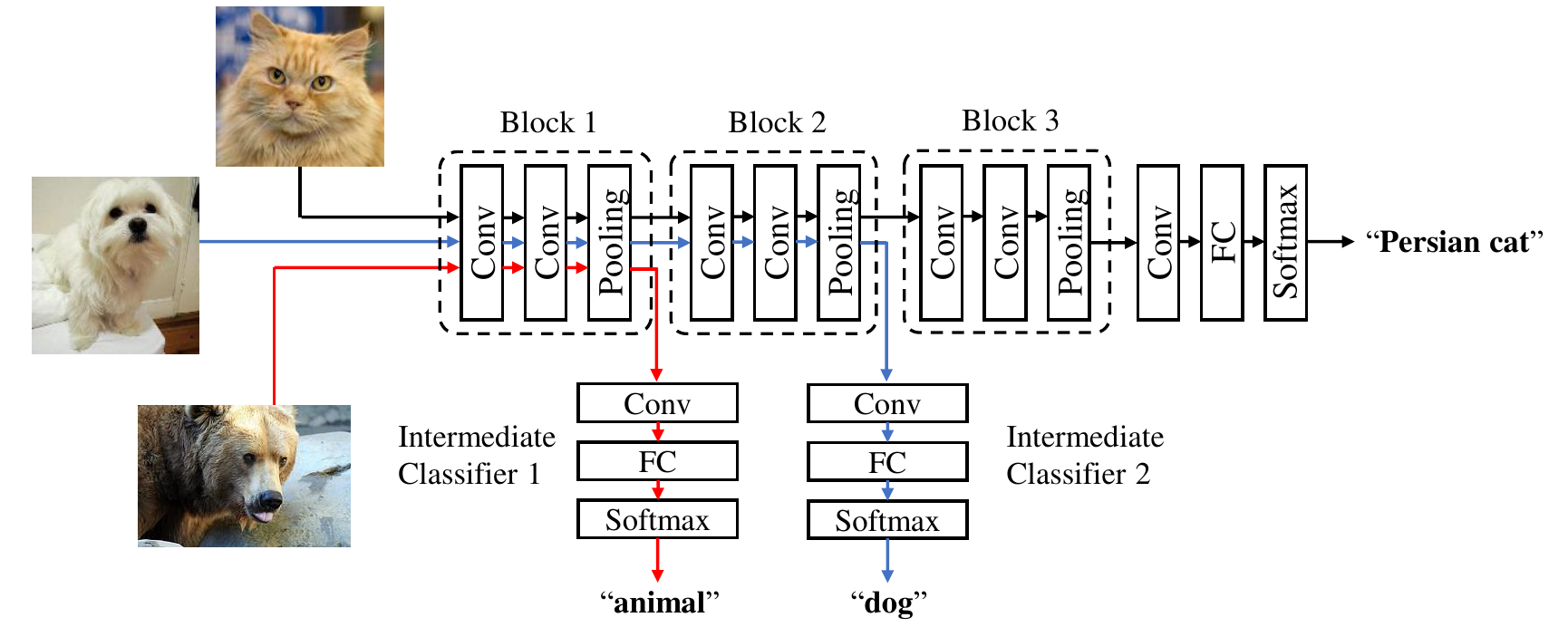}
  \caption{A three-block deep neural network with early exits. }\label{fig_diagram_earlyexit}
  \centering
\end{figure*}

In the considered context of deep learning, the classification accuracy of an inference model refers to the percentage of correctly classified samples within the entire testing dataset. As a common practice in the machine learning literature, this empirical accuracy approximates the probability that an uploaded task sample is correctly classified \cite{goodfellow2016deep}. Full-network inference provides the same accuracy to all tasks. In the early exiting case, each task is specified with the target accuracy $a_k$ with $0<a_k<1$. Prior to the deployment of the edge-inference system, the server profiles the deterministic monotonic increasing relationship between the empirical inference accuracy and the number of executed blocks (see the examples of GoogLeNet, ResNet-50 and MSDNet in \cite{pmlr-v70-bolukbasi17a,MSDNet}). Using the relation, the server can determine the minimum number of blocks for a task, say task $k$, that need be executed to meet the target accuracy $a_k$. In other words, the exit points $\{d_k\}$ are known. 

\section{Problem Formulation}
\label{sec: problem_formulation}
The objective of optimal resource allocation is to maximize the number of completed multiuser tasks within each epoch under their heterogeneous QoS constraints. In this section, the optimal designs are  formulated as two optimization problems for the full-network inference and early exiting case respectively. Let $\cS \subseteq \cK$ denote the subset of selected tasks in one epoch, and $\vert \cS  \vert$ its cardinality. Then the design objective is to maximize $|\cS|$. By allocating to each selected task exactly the minimum bandwidth it requires, the bandwidth constraint is  $\sum_{k\in \cS} \rho_{\min, k} \leq 1$. First, consider the full-network inference case without early exiting. On one hand, if task $k$ is selected, its end-to-end latency constraint must be met, i.e., $t_{\mathsf{wt},k}+t_{\mathsf{ul},k}+t_{\mathsf{cp},k}\leq \tau_k$. On the other hand, all inference computation must be finished within the duration of the computation slot, which gives $t_{\mathsf{cp},k}\leq T$ for all $k\in\cS$. Consequently, the constraints on end-to-end latency and inference latency can be merged into $f(|\cS|)\leq \tilde{\tau}_k$ for all $k \in \cS$, where $\tilde{\tau}_k\triangleq \min\{\tau_k-t_{\mathsf{wt},k}-T, T\}$ is the required computation latency for task $k$, and referred to as the \emph{latency budget} in the sequel. Under the above constraints, the resource allocation problem for the full-network inference case is formulated as follows: 
\begin{equation*}\text{(P1)}\quad\quad\quad 
\begin{aligned}
\max\limits_{\cS\subseteq \cK}\quad & |\cS|  \\
   \mathrm{s.t. }\quad & \sum_{k\in \cS} \rho_{\min, k} \leq 1, \\
                       & f(|\cS|)\leq \tilde{\tau}_k, \quad \forall k \in \cS.
\end{aligned}
\end{equation*}
{Next, consider the early-exiting case where accuracy requirements can be heterogeneous. Diversified accuracy requirements are translated into task-specific constraints on computation latency  as elaborated shortly.} The batch size at block $d$ is the number of selected tasks whose exit points are after block $d$, i.e., $n_d=\sum_{k\in\cS}\mathbb{I} {\left(d_k\geq d\right)}$, where $\mathbb{I}{(\cdot)}$ represents the indicator function. As a result, the latency constraints are recast to formulate the following optimization problem for the early exiting case: 
\begin{equation*}\text{(P2)}\quad\quad\quad 
\begin{aligned}
\max\limits_{\cS\subseteq \cK}\quad & |\cS|  \\
   \mathrm{s.t. }\quad & \sum_{k\in \cS} \rho_{\min, k} \leq 1, \\
                       & \sum_{d=1}^{d_k} f_d(n_d)\leq \tilde{\tau}_k, \quad \forall k \in \cS.
\end{aligned}
\end{equation*}

The combinatorial nature of the resource allocation problems renders the optimal policy design challenging. 
%An exhaustive search over all subsets of $\cK$ is intractable due to the prohibitive complexity. 
In particular, the following result suggests that no polynomial-time algorithm exists for Problem~P2 unless P=NP. 
\begin{proposition}
\emph{Problem~P2 is NP-complete}. 
\end{proposition}
\begin{proof}
The two-dimensional knapsack problem with the target of maximizing cardinality aims to contain a maximum number of items into a single knapsack with capacity limits in two dimensions. Specifically, let $\cK^\dagger=\{1,\ldots,K^\dagger\}$ be the index set of all items. The weight and volume of item $k$ are denoted by $a_k$ and $b_k$, respectively. The problem of selecting the maximum number of items from $\cK^\dagger$ while not exceeding the weight and volume capacities of the knapsack, denoted by $C_1$ and $C_2$ respectively, is mathematically formulated as 
\begin{equation*}
\begin{aligned}
\max & \sum_{k=1}^{K^\dagger} x_k  \\
   \mathrm{s.t. }\quad & \sum_{k=1}^{K^\dagger} a_k x_k\leq C_1, \\
                       & \sum_{k=1}^{K^\dagger} b_k x_k \leq C_2, \\
                       & x_k\in\{0,1\},
\end{aligned}
\end{equation*}
where $x_k=1$ if the item with index $k$ is selected and vice versa. The items of the above problem can be mapped to the tasks of Problem~P2 by letting $\cK=\cK^\dagger$ and $\cS = \{k|x_k=1\}$, and thus the target of the above problem is equivalent to that of Problem~P2. The weight constraint is also equivalent to the bandwidth constraint of Problem~P2, as the item weights can be mapped to the bandwidth requirements of Problem P2 by letting $\rho_{\min,k}=a_k/C_1$. Next, we show that the volume constraint is a special case of the latency constraint of Problem~P2. To this end, we construct a task set for Problem~P2 with $d_k=k$ and $\tilde{\tau}_k=C_2$ for $k=1,\ldots,K^\dagger$ and $D=K^\dagger$. Hence, we have $n_d = \sum_{k=d}^{K^\dagger} x_k$. Further, noting that $f_d(\cdot)\geq 0$, the latency constraint is equivalent to $\sum_{d=1}^{k_{\max}} f_d(n_d)\leq C_2$, where $k_{\max} \triangleq \max\limits_{k} \{k|x_k=1\}$. Consider a particular linear batch computation time model where $f_1(n_1) = b_1 n_1$ and $f_d(n_d)=(b_d-b_{d-1})n_d$ for $d=2,\ldots,D$. Without loss of generality, we assume that $b_1\leq b_2 \leq \ldots\leq b_{K^\dagger}$ such that $f_d(n_d)$ is monotonically increasing. The latency constraint then becomes
\begin{equation*}
    \sum_{d=1}^{k_{\max}} b_d x_d = \sum_{k\in\cS} b_k \leq C_2,
\end{equation*}
which is equivalent to the volume constraint. Therefore, we  conclude that the cardinality-maximizing two-dimensional knapsack problem, which according to \cite{KORTE1981} is NP-complete, is a special case of Problem~P2. Hence, Problem~P2 is NP-complete. This completes the proof. 
\end{proof}

To overcome these challenges, in the following sections, we first propose a low-complexity optimal algorithm for solving Problem P1, and then modify this algorithm to develop a sub-optimal algorithm for solving Problem P2. Furthermore, via analyzing the solution for Problem P2 in the special large-bandwidth case, we exploit the insight to develop an efficient tree-search algorithm for finding the optimal solution.

\section{Multiuser Edge Inference without Early Exiting}
\label{sec: MEAIwoEE}

In this section, we consider the case where users have similar QoS requirements and thus early exiting is not necessary, i.e., the full-network inference case. For this case, the optimal resource allocation for multiuser edge inference is designed by solving Problem~P1. First, we transform the problem into a series of feasibility problems. Then, a solution  algorithm is designed  by solving these feasibility problems sequentially and its optimality is proved. 

To solve Problem~P1, we first analyze the latency constraint, i.e., $f(\vert\cS\vert)\leq \tilde{\tau}_k$ for any $k \in \cS$. The batch computation latency $f(\vert\cS\vert)$ depends on the cardinality but not particular elements of $\cS$. Hence, given the cardinality fixed as  $\vert \cS \vert=n$, the latency constraint can be simplified as $f(n)\leq \tilde{\tau}_k$. This implies a simple threshold structure on the latency feasibility that, any task whose latency budget $\tilde{\tau}_k$ satisfies $\tilde{\tau}_k \geq f(n)$ is feasible conditioned on $n$ tasks batched. This motivates us to study the following feasibility problem: \emph{Does there exist $n$-task subsets that satisfy the bandwidth and latency constraints?} The problem of finding such subsets for $n$ can be mathematically written as
\begin{equation*}\text{(P3)}\quad\quad\quad
\begin{aligned}
\mathrm{find} \quad & {\cS\subseteq \cK}   \\
   \mathrm{s.t. }\quad & |\cS| = n, \\
                       & \sum_{k\in \cS} \rho_{\min, k} \leq 1, \\
                       & f(n)\leq \tilde{\tau}_k, \quad \forall k \in \cS. 
\end{aligned}
\end{equation*}
The threshold $f(n)$ divides the set of all requested  tasks $\cK$ into $\cF$ and the complement $\cK\setminus\cF$, where $\cF$, referred to as the \emph{feasible subset} of $\cK$, consists  of all tasks whose latency budget is not lower than $f(n)$. Mathematically, 
\begin{equation}\label{eqn:no_ee_feasible_subset}
    \cF =\{ k|k\in\cK, f(n)\leq\tilde{\tau}_k \}.
\end{equation} 
The collection of $n$-element task subsets that do not violate the latency constraint is then simply all $n$-element subsets of $\cF$. Consequently, if the solution $\cS$ exists, it must be an $n$-element subset of $\cF$. If $|\cF|< n$, obviously Problem~P3 has no solution. On the other hand, if $|\cF|\geq n$, to solve Problem~P3 only requires  finding an $n$-element subset of $\cF$ that satisfies the bandwidth requirement, i.e., $\cS\subseteq\cF$ such that $|\cS| = n$ and $\sum_{k\in \cS} \rho_{\min, k} \leq 1$. To this end, we sort the tasks in $\cF$ in  ascending order according to the bandwidth requirement $\rho_{\min, k}$ and let $\cS$ contain the first $n$ tasks since they require  the smallest bandwidths. If the sum bandwidth requirement of $\cS$ is smaller than or equal to $1$, $\cS$ satisfies  both the bandwidth and latency constraints and hence constitutes a solution for  Problem~P3. Otherwise, we can conclude that  Problem~P3 is infeasible.

We have so far solved Problem~P3 for any given $n$. To solve Problem~P1, we {sequentially} solve P3 for $n=1, 2, \cdots$, until we find the smallest $n^\star$ such that $(n^\star+1)$ makes Problem~P3 infeasible.
 Then the corresponding solution for P3 with $n^\star$ also solves Problem~P1. We summarize the algorithm in Algorithm~\ref{algorithm: opt_p1} with optimality proved in the sequel. 

\begin{algorithm}[t]
\caption{Best-shelf-packing Algorithm for Solving Problem~P1}
\label{algorithm: opt_p1}
\textbf{Input:} The task index set $\cK=\{1,2,\ldots,K\}$, the bandwidth requirements $\{\rho_{\min, k}\}_{k=1}^{K}$, the latency constraints $\{\tilde{\tau}_k\}_{k=1}^{K}$;\\
\textbf{Initialize:} $\cS^{(0)}=\emptyset$;\\
Sort $\cK$ in ascending order according to $\rho_{\min, k}$; \\
\textbf{for} $n=1,2,\cdots,K$ \textbf{do}\\
    \begin{enumerate}
        \item[] $\cF \leftarrow \{ k|k\in\cK, f(n)\leq\tilde{\tau}_k \}$;\\
                \textbf{if} $|\cF|<n$ \textbf{then return} $\cS^{(n-1)}$; \\
                \textbf{else}\\
                \begin{itemize}
                    \item[]  Let $\cS'$ contain $n$ tasks of $\cF$ with minimum bandwidth requirements; 
                    \item[] \textbf{if} $\sum_{k\in \cS'} \rho_{\min, k} \leq 1$ \textbf{then}\\
                    \begin{itemize}
                        \item[] \emph{Problem~P3 is feasible given $n$: }$\cS^{(n)}\leftarrow\cS'$;
                    \end{itemize}
                    \textbf{else return} $\cS^{(n-1)}$
                \end{itemize}
    \end{enumerate}
\textbf{end for} \\    
\textbf{return} $\cS^{(K)}$\\
\end{algorithm}

\begin{proposition}
\emph{Algorithm~\ref{algorithm: opt_p1} is optimal for Problem~P1.}
\end{proposition}
\begin{proof}
Given finding $n^\star$ that Problem~P3 is feasible for $n=n^\star$ with solution $\cS^{(n^\star)}$ and infeasible for $n=n^\star+1$, Problem~P3 is infeasible for any $n>n^\star+1$, which can be easily proved by contradiction. This means that the objective of Problem~P1 is upper bounded by $n^\star$, i.e., $|\cS|\leq n^\star$. Meanwhile, there exists a candidate point $\cS^{(n^\star)}$ satisfying all constraints with the objective $|\cS^{(n^\star)}|=n^\star$. Therefore, $\cS^{(n^\star)}$ is the optimal solution to Problem~P1. 
\end{proof}

The complexity of Algorithm~\ref{algorithm: opt_p1} is $\cO(K^2)$ as explained below. The complexity of sorting $\cK$ is $\cO(K\log K)$. In each iteration,  identifying $\cF$ and computing $\sum_{k\in \cS} \rho_{\min, k}$ both have complexity $\cO(K)$. In the worst case, the loop is repeated $K$ times, yielding the said complexity of Algorithm~\ref{algorithm: opt_p1}.

\begin{figure*}[!t]
  \centering
  \includegraphics[width=0.95\textwidth]{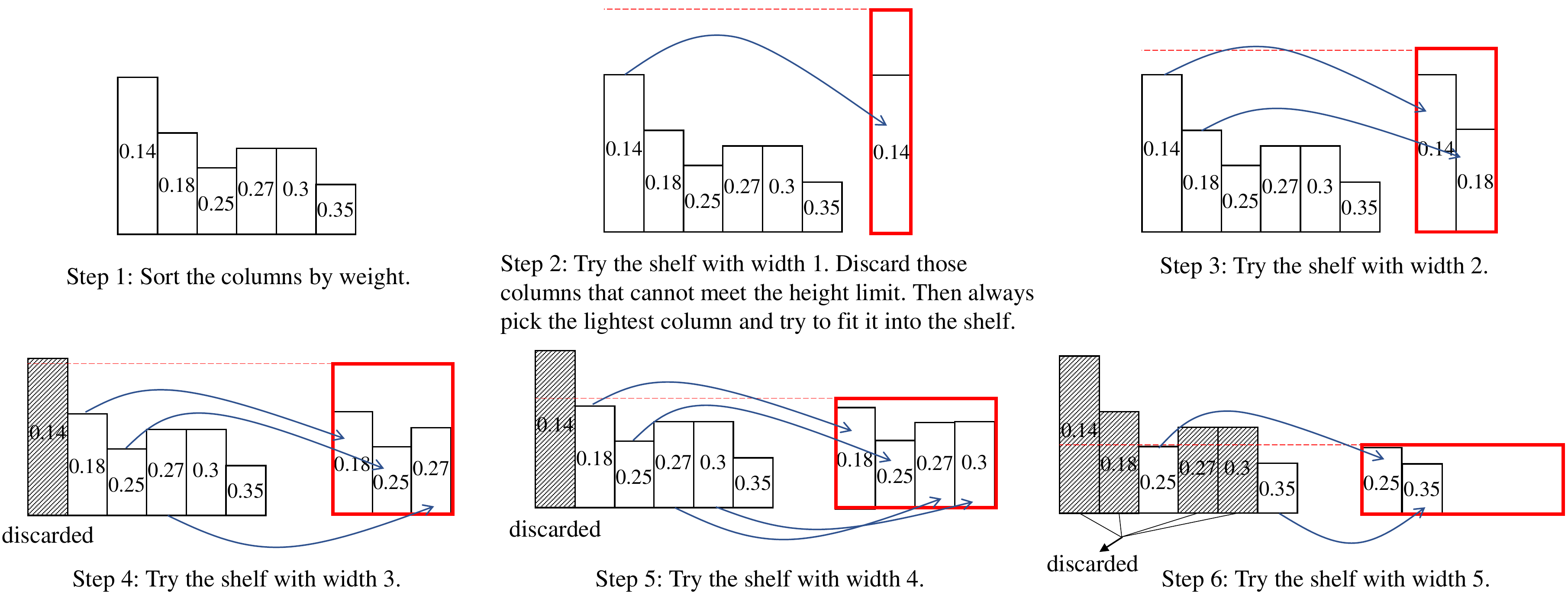}
  \caption{Shelf packing interpretation of Algorithm~\ref{algorithm: opt_p1}. Step 6 finds no packing of 5 columns, and further increasing the box width will not yield better solutions. Hence, Step 5 yields the optimal solution.  }\label{fig_alg1}
  \centering
\end{figure*}

\begin{remark}(Geometric Interpretation). 
 \emph{Algorithm~\ref{algorithm: opt_p1} can be interpreted as a packing procedure by considering the linear batch computation time model, i.e., $t_\mathsf{cp}=\hat{f}(n_\mathsf{b})=c_0 + c_1 n_\mathsf{b}$, where $c_0$ and $c_1$ are positive constants\cite{Inoue2021}. As shown in Fig.~\ref{fig_alg1}, each task can be viewed as a 2-D column with unit width, height $y_k=\frac{T-\tilde{\tau}_k}{c_1}$ and weight $\rho_{\min, k}$. Note that $y_k$ is a decreasing function of the associated latency budget, $\tilde{\tau}_k$. In other words, a looser latency constraint leads to a lower height. The latency constraint can be translated to  $y_k\leq \frac{T-c_0}{c_1}-|\cS|$. Hence, Problem~P1 is equivalent to packing as many columns as possible into a 2-D rectangular shelf with fixed perimeter $P\triangleq\frac{T-c_0}{c_1}$ under the constraint that   the columns'  sum weight is no larger than  $1$. Such a packing is termed a \emph{shelf packing}. As illustrated in Fig.~\ref{fig_alg1}, Algorithm~\ref{algorithm: opt_p1}, a best-shelf-packing algorithm increases the shelf width by $1$ in each iteration. In this process, the height constraint becomes more stringent such that  only columns with lower heights can be fitted into the shelf.  This is consistent with the intuition that a large batch size delays the exit point, and thus is suitable for tasks with loose latency requirements. The algorithm stops if no more columns can be packed into the shelf.} 
\end{remark}

\section{Multiuser Edge Inference with Early Exiting}
\label{sec: MEAIwEE}

In this section, we consider the case where users have heterogeneous  QoS requirements and thus early exiting is required. Recall that early exiting can accommodate low-latency tasks  due to the early-feedback effect and support more efficient computation by the effect of shrinking batch size. The corresponding optimal resource allocation for the current case is designed by solving Problem~P2. With early exits, Problem~P2 is a more complex problem to solve than Problem P1, the full-network inference counterpart. Algorithm~\ref{algorithm: opt_p1} for the latter, however, can be modified to develop a  low-complexity  algorithm for finding a sub-optimal solution for the former. Next, we pursue the optimal solution for Problem~P2.  {The resource allocation is divided into two cases with large and limited bandwidth,  respectively. In the former case, the bandwidth is always sufficient and its constraint is thus inactive. Consequently, the computation bottleneck dominates the communication counterpart such that the main challenge for optimal  control lies in meeting the heterogeneous requirements on computation latency and accuracy from  multiple tasks. On the other hand, the limited-bandwidth case is made more complex as the communication bottleneck also comes into play and communication latency should be accounted for. } By analyzing the former case, we obtain useful insights to propose an optimal tree-search algorithm for Problem~P2 featuring depth-first search and node pruning.

\subsection{Sub-optimal Resource Allocation}
Consider the feasibility problem, Problem~P3,  conditioned on $n$ scheduled tasks. It can be modified for the current  early exiting case as 
\begin{equation*}\text{(P4)}\quad\quad\quad
\begin{aligned}
\mathrm{find} \quad & {\cS\subseteq \cK}   \\
   \mathrm{s.t. }\quad & |\cS| = n, \\
                       & \sum_{k\in \cS} \rho_{\min, k} \leq 1, \\
                       & \sum_{d=1}^{d_k} f_d(n_d)\leq \tilde{\tau}_k, \quad \forall k \in \cS, 
\end{aligned}
\end{equation*}
where $n_d=\sum_{k\in\cS}\mathbb{I}_{d_k\geq d}$. Problems P3 and P4 differ in their last constraints on latency. Given that the number of selected tasks is fixed as $|\cS| = n$, the last constraint in the former depends only on $n$ but that in the latter is complex and depends on exit points $\{d_k\}$ for individual tasks in $\cS$. As a result, the threshold structure of  the latency constraint for the full-network inference case no longer holds in the current case, preventing the  direct application  of  Algorithm~\ref{algorithm: opt_p1} to the early exiting case. However, a sub-optimal solution for Problem P4 can be developed by tightening the latency constraint to facilitate the use of Algorithm~\ref{algorithm: opt_p1}. Specifically, given $|\cS|=n$, one can observe that
\begin{equation}
    n_d=\sum_{k\in\cS}\mathbb{I}_{d_k\geq d}\leq n, \forall d=1,2,\ldots,D,
\end{equation}
where the equality holds for all $d=1,2,\ldots,D$ if and only if $d_k= D$ for any $k\in \cS$.
Since the function  $f_d(n_d)$ is non-decreasing, the  latency constraint can be  tightened as 
\begin{equation}
    \sum_{d=1}^{d_k} f_d(n)\leq \tilde{\tau}_k,\quad \forall k \in \cS, 
\end{equation}
to have a similar form as that of Problem~P3. The tightened constraint also has the mentioned threshold structure as that in Problem~P3 despite a different form depending on the exit point $d_k$ and thus varying  over tasks. We can then write the feasible task subset as
\begin{equation}
    \label{eqn: suboptimal_feasible}
    \cF'\triangleq\left\{ k|k\in\cK, \sum\nolimits_{d=1}^{d_k} f_d(n)\leq\tilde{\tau}_k \right\},
\end{equation}
of which all the $n$-element subsets will satisfy the tightened latency constraints and accuracy requirements in Problem P4. Consequently, the sub-optimal algorithm for P2 can be obtained by modifying the feasible-subset-finding step in Algorithm~\ref{algorithm: opt_p1} by replacing $\cF$ in~\eqref{eqn:no_ee_feasible_subset} with  $\cF'$ in \eqref{eqn: suboptimal_feasible}. The modified algorithm then evaluates the sum bandwidth of $n$ tasks  in the feasible set, which are identified  using the minimum-bandwidth criterion,  to determine the feasibility of Problem~P4, and repeats solving Problem~P4 sequentially for $n=1,\ldots,K$ to find the maximum feasible value of $n$ and the corresponding task subset $\cS$.

\begin{remark} (Effects of Constraint Tightening)\emph{
 As we tighten the latency constraint by replacing $n_d$ with its upper bound $n$ earlier, it is assumed  that the batch size does not shrink as completed tasks early exit, therefore overestimating the computation latency. On the other hand, despite the assumption, early exiting is still implemented with inference results immediately fed back to users upon reaching their  corresponding exits. In other words, the resource allocation policy given constraint tightening enables edge inference to  benefit from the early-feedback effect while the effect of shrinking batch size is ignored.}
\end{remark}

\subsection{Optimal Resource Allocation with Large Bandwidth}\label{section: infy_bw}

In the rest context, we pursue the optimal algorithm for multiuser edge inference with early exiting, i.e., optimally solving Problem~P2. {To solve this  NP-hard problem with two strongly coupled constraints (i.e., the bandwidth and latency constraints), we propose the following tractable solution approach. Recall that solving Problem~P2 can be transformed into nesting solving Problem~P4, the associated   feasibility problem, into a sequential search. The key to solving Problem~P4 is its  decomposition into two steps.  First, we relax the bandwidth constraint to study the structure of feasible solutions when only the latency constraint exists, i.e., in the large-bandwidth case. This results in an optimal greedy user selection scheme for this case. The second step is inspired by the idea that the solution space of the large-bandwidth case can be enumerated to identify feasible solutions to the limited-bandwidth case, i.e., those also satisfying the bandwidth constraint. Materializing the idea leads to the design of a practical online-tree-search algorithm that optimally solves Problem~P4, with its complexity reduced by developing a tree-pruning technique.} 

To start with, we first consider the case of large bandwidth, { i.e., where the bandwidth $B$ is sufficiently large such that the bandwidth constraint is always inactive and can be thus omitted. Mathematically, a sufficient condition is that for any $\mathcal{S}\in\mathcal{K}$ that satisfies the latency constraint, the following inequality holds:
\begin{equation}
    B \geq \sum_{k\in\mathcal{S}}\frac{\ell_k}{{T \log_2(1+p_k h_k^2/N_0)}}.
\end{equation}}
For ease of exposition, we split the task set $\cK=\{1,2,\ldots,K\}$ into $D$ distinct groups $\cK_1, \cK_2, \ldots, \cK_D$, where the $m$-th group comprises tasks that exit right after traversing the $m$-th block, i.e., $\cK_m \triangleq \{k|k\in\cK, d_k=m\}$. The set of selected tasks $\cS$ can also be split into $\cS_1, \cS_2, \ldots, \cS_D$ according to the exit points, where $\cS_m \triangleq \{k|k\in\cS, d_k=m\}$ contains all selected tasks whose exit points equal to $m$. Using this notation, the block-wise batch size $n_d$, i.e., the number of selected tasks that need to traverse at least $d$ blocks, can be written as
\begin{equation} \label{equation: n_d}
\begin{aligned}
    n_d & = \sum_{k\in\cS}\mathbb{I}_{d_k\geq d} = \sum_{m=d}^D\sum_{k\in\cS}\mathbb{I}_{d_k= m} = \sum_{m=d}^D |\cS_m|.
\end{aligned}
\end{equation}
Problem~P4 is equivalent to selecting $\cS_m$ from each $\cK_m$ such that $\sum_{m=1}^D |\cS_m| = n$ while not violating the bandwidth and latency constraints. In the large-bandwidth case, the bandwidth constraint is omitted and the feasibility problem conditioned on $n$ can thus be formulated as
\begin{equation*}\text{(P5)}\quad\quad\quad
\begin{aligned}
\mathrm{find} \quad & \cS=\bigcup_{m=1}^D \cS_m   \\
   \mathrm{s.t. }\quad & \sum_{m=1}^D |\cS_m| = n, \\
                       & \sum_{d=1}^m f_d(n_d)\leq \tilde{\tau}_k, \; \forall k \in \cS_m, m=1,\ldots,D. 
\end{aligned}
\end{equation*}
Each task then has only two attributes under consideration: the latency budget $\tilde{\tau}_k$ and the exit point $d_k$. 

\subsubsection{Optimal Greedy User Selection}
Note that with $n_1=\sum_{m=1}^D|\cS_m| = n$, the latency constraint for tasks in $\cS_1$, i.e., those selected from $\cK_1$, becomes $f_1(n)\leq \tilde{\tau}_k$, which has a threshold structure, and therefore divides $\cK_1$ into $\cF_1$ and $\cK_1\setminus\cF_1$. Here, $\cF_1$ denotes the feasible subset of $\cK_1$ consisting all task in $\cK_1$ whose latency requirement is beyond the threshold $f_1(n)$, and is given as 
\begin{equation}
    \cF_1=\{k|k\in\cK_1,f_1(n)\leq \tilde{\tau}_k\}.
\end{equation}
It then follows that $\cS_1$ must be a subset of $\cF_1$. If $|\cF_1|\geq n$, we can choose $\cS_1$ as any $n$-element subset of $\cF_1$, and conclude that Problem~P5 is feasible with solution $\cS=\cS_1$. If $|\cF_1|< n$, then $\cS_1$ can be any subset of $\cF_1$, leading to many candidate searching branches. However, the following lemma suggests that only one branch needs to be explored, which is proven in Appendix A.
\begin{Lemma}[\emph{Optimal Greedy User Selection}]
\emph{
For Problem~P5, if $\vert\cF_1\vert< n$ and no solution exists such that $|\cS_1|=n^\dagger\leq|\cF_1|$, then there exists no solution such that $0\leq|\cS_1|<n^\dagger$. As a result, it is optimal to set $\cS_1=\cF_1$. }
\label{lemma: local_optimality}
\emph{} 
\end{Lemma}
From Lemma~\ref{lemma: local_optimality}, we can assert that given $|\cF_1|< n$, if no solution exists such that $\cS_1=\cF_1$, then Problem~P5 is infeasible and it is unnecessary to search for solutions such that $\cS_1\subsetneq\cF_1$. It follows that the \emph{greedy user selection} should be adopted for selecting $\cS_1$, i.e., $\cS_1=\cF_1$ as stated in the lemma. The intuition behind greedy selection is given as follows. We observe that tasks in $\cK_1$ exit after traversing only one block, no longer contributing to the batch size of subsequent blocks. Hence, to maximize the effect of shrinking batch size, the tasks in $\cK_1$ should be preferably selected, which motivates the selection of the entire feasible subset of $\cK_1$ as  $\cS_1$.

\subsubsection{Nested Sub-problems}

Given $\cS_1=\cF_1$, the next question is: \emph{Can we find $\{\cS_2, \ldots, \cS_D\}$ such that $\sum_{m=2}^D |\cS_m| = n-|\cS_1|$ while meeting the latency constraints?} The sub-problem can be generalized to be that given  $\cS_1,\ldots,\cS_{a-1}$,  find $\{\cS_a,\ldots,\cS_D\}$   such that $\sum_{m=a}^D |\cS_m| = n-\sum_{m=1}^{a-1} |\cS_m|$. Mathematically, 
\begin{equation*}\text{(P5($a$))}\quad
\begin{aligned}
\mathrm{find} \quad & \cS=\bigcup_{m=a}^D \cS_m   \\
   \mathrm{s.t. }\quad & \sum_{m=a}^D |\cS_m| = n', \\
                       & \sum_{d=a}^m f_d(n_d)\leq \tilde{\tau}'_k,\; \forall k \in \cS_m, m=a,\ldots,D. 
\end{aligned}
\end{equation*}
In Problem P5($a$), $n'=n_a=n-\sum_{m=1}^{a-1} |\cS_m|$ is the number of tasks the sub-problem aims to select and $\tilde{\tau}'_k=\tilde{\tau}_k-\sum_{d=1}^{a-1}f_d(n_d)$ is the remaining latency budget, given that a time duration of $\sum_{d=1}^{a-1}f_d(n_d)$ is needed to compute the first $a-1$ blocks. Note that $\{\tilde{\tau}'_k\}$ is constant in Problem~P5($a$) because $n_d=n-\sum_{m=1}^{d-1}\vert\cS_m\vert$ is fixed for $d=1,\ldots,a$ given that $\cS_1,\ldots,\cS_{a-1}$ has been selected. 
The above sub-problem, Problem~P5($a$), then has the same structure as Problem~P5, except that it has reduced dimension. Most importantly, for selection of $\cS_a$, the latency constraint has a threshold structure: $f_a(n')\leq \tilde{\tau}'_k, \forall k \in \cS_a$.  Consequently, the feasible subset of $\cK_a$, denoted by $\cF_a$, is given by $\cF_a=\{k|k\in\cK_a,f_a(n')\leq \tilde{\tau}'_k\}$. If $|\cF_a|\geq n'$, then Problem~P5($a$) is feasible by choosing $\cS_a$ to be an arbitrary $n'$-element subset of $\cF_a$. Otherwise, by Lemma~\ref{lemma: local_optimality} we greedily choose $\cS_a$ to be the entire feasible subset of $\cK_a$, i.e., $\cS_a=\cF_a$, and then a new sub-problem, Problem~P5($a+1$), arises for choosing $\{\cS_{a+1},\ldots,\cS_D\}$ in the case of $a<D$. It turns out that a series of \emph{nested} sub-problems subsequently arises, i.e., P5($a=2$), P5($a=3$), ..., P5($a=D$). Moreover, the following proposition on the feasibility of these sub-problems holds for $a=1,\ldots,D-1$, where Problem~P5($a=1$) refers to Problem~P5 for ease of notation.
\begin{proposition}\label{prob: sub_prob_infeas}
\emph{If $|\cF_a|< n'$, the feasibility of Problem~P5($a+1$) implies that of Problem~P5($a$).}
\end{proposition}
\begin{proof}Obviously, Problem~P5($a$) is feasible if Problem~P5($a+1$) is feasible. If Problem~P5($a+1$) is infeasible, then there exists no solution for Problem~P5($a$) such that $\cS_a=\cF_a$. Consequently, by Lemma~\ref{lemma: local_optimality}, there exists no solution such that $\cS_a\subset\cF_a$ either, which proves the infeasibility of Problem~P5($a$). This completes the proof. 
\end{proof}

In particular, the innermost sub-problem, P5($a=D$), is determined to be infeasible if $|\cF_D|<n'=n-\sum_{m=1}^{D-1} |\cS_m|$.  

\subsubsection{Recursive Algorithm}
Given the structure of nested sub-problems, Problem~P5 can be solved recursively. The recursive algorithm starts from Problem~P5($a=1$). For $a=1,\ldots,D-1$, Problem~P5($a$) either immediately returns a feasible solution when $|\cF_a|\geq n'$, or, when $|\cF_a|< n'$, greedily selects $\cS_a=\cF_a$ and enters Problem~P5($a+1$). On one hand, during the recursive process, if Problem~P5($a$) is found to be feasible, then Problem~P5($a-1$) is feasible by Proposition~\ref{prob: sub_prob_infeas}, and thus Problem~P5 is feasible by induction. The algorithm can then stop and return $\cS=\bigcup_{m=1}^a \cS_m$ as a feasible solution to Problem~P5. On the other hand, if the algorithm reaches Problem~P5($a=D$), but determines Problem~P5($a=D$) to be infeasible, then by induction Problem~P5 is infeasible. The detailed algorithm  is presented  in Algorithm~\ref{algorithm: greedy}.

{\subsubsection{Complexity Analysis}
A single run of Algorithm~\ref{algorithm: greedy} involves at most $D$ iterations. The computational complexity of each iteration is dominated by that of finding $n'$, $\tilde{\tau}_k'$ and $\cF_d$, and thus given by $\cO(\max\{K,D\})$. The overall complexity of Algorithm~\ref{algorithm: greedy} is given by $\cO(D\max\{K,D\})$. }

\begin{algorithm}[t]
\caption{Greedy-user-selection Algorithm for Solving  Problem~P5}
\label{algorithm: greedy}
\textbf{Input:} The task index set $\cK=\{1,2,\ldots,K\}$, the bandwidth requirements $\{\rho_{\min, k}\}_{k=1}^{K}$, the latency constraints $\{\tilde{\tau}_k\}_{k=1}^{K}$, the required exit points $\{d_k\}_{k=1}^{K}$ and the desired number of tasks $n$;\\
\textbf{Initialize:} $\cK_m \triangleq \{k|k\in\cK, d_k=m\}$, $\cS_m=\emptyset$ for $m=0,\ldots,D$;\\
\textbf{for} $a=1,2,\cdots,D$ \textbf{do}\\
    \begin{enumerate}
        \item[] Find $n'$ and $\tilde{\tau}_k'$ for Problem~P5($a$);
        \item[] $\cF_a\leftarrow \{k|k\in\cK_a, f_a(n')\leq \tilde{\tau}_k'\}$;
        \item[] \textbf{if} $|\cF_a| \geq n'$ \textbf{then}\\
        \begin{enumerate}
            \item[] $\cS_a\leftarrow$ any $n'$-element subset of $\cF_a$;
            \item[] \textbf{return} $\cS=\bigcup_{m=1}^a \cS_m$;
        \end{enumerate}
        \item[] \textbf{else} $\cS_a\leftarrow\cF_a$; 
    \end{enumerate}
\textbf{end for} \\
\textbf{return} no solution\\    
\end{algorithm}

\subsection{Optimal Resource Allocation with Limited  Bandwidth}
Consider Problem~P4 for the current case with limited bandwidth conditioned on $n$ scheduled task, which is equivalent to Problem P5 in the preceding sub-section added with the bandwidth constraint. Due to the extra constraint, solving the former requires a more sophisticated approach than  Algorithm~\ref{algorithm: greedy}. Given the combinatorial nature of Problem~P4, we adopt the solution method of tree-search.  Algorithm~\ref{algorithm: greedy} that solves Problem P5 is used as a mechanism for checking the feasibility of meeting the latency constraint so as to reduce the search complexity. In the sequel, the search tree is constructed, followed by the design of an efficient search algorithm integrating depth-first search and node pruning.

\subsubsection{Search-tree Construction}
As shown in Fig.~\ref{fig_trees_1}, a search tree originates from one root node, denoted by $\bv_0$. The definitions of the root node, general nodes and the steps of \emph{visiting} a node are given as follows. 

\begin{enumerate}
    \item[(a)] \emph{Root node.} The root node is represented by an empty vector $\bv_0=[\;]$. The child nodes of the root node, each representing one possible choice of $\cS_1$ that satisfies the latency constraint, form the set of nodes with depth $1$. Naturally, $\cS_1$ can be any subset of the feasible subset $\cF_1$ as defined in \eqref{eqn: suboptimal_feasible}, which leads to $2^{|\cF_1|}$ choices. However, tasks in $\cF_1$ are homogeneous in terms of the latency requirement $\tilde{\tau}_k$ and accuracy requirement $d_k$. Thus priorities are given to tasks with lower bandwidth requirements. As a result, given its cardinality, $\cS_1$ automatically consists of $\vert\cS_1 \vert$ elements with minimum bandwidth requirements from $\cF_1$. This greatly reduces the number of branches as each branch now represents the cardinality of $\cS_1$ with significantly fewer choices. Specifically, $|\cS_1|$ can be any value between $0$ and $|\cF_1|$ when $\vert\cF_1\vert\leq n$, and when $\vert\cF_1\vert>n$, $\vert\cS_1\vert$ cannot exceed $n$. Hence, we need $M+1$ child nodes indexed by $0,1,\ldots,M$, where $M=\min\{n,\vert\cF_1\vert\}$, to represent all possible searching branches with the node index indicating the size of $\cS_1$. 
    
    \item [(b)] \emph{General nodes.} The depth of an arbitrary node $\bv$, denoted by $d(\bv)$, is the length of the path connecting it to the root node. As a general rule, the index of a node with depth $m$ indicates the size of $\cS_m$, which is unique among all its sibling nodes. We can then uniquely refer to a node $\bv$ with a vector $\bv=[v_1,v_2,\ldots,v_{d(\bv)}]$ to record the indices of nodes along the path from the root node to node $\bv$, where $v_{d(\bv)}$ denotes the index of $\bv$ itself. Hence, a general node $\bv=[v_1,v_2,\ldots,v_{d(\bv)}]$ defines a partial solution with $\vert \cS_m\vert=v_m$, where $m=1,\ldots,d(\bv)$. Note that such a partial solution is guaranteed to satisfy the latency constraint because $\cS_m$ is always selected out of the feasible subset of $\cK_m$.
    
    \item [(c)] \emph{Visiting a node.} To determine the properties of node $\bv$ and discover its potential child nodes, we need to \emph{visit} node $\bv$ by executing  the following operations.  We determine $\bv$ to be a \emph{leaf node} without child nodes if it satisfies either of the two following conditions: (1) $\sum_{m=1}^{d(\bv)}v_m=n$, in which case the accumulated number of selected tasks reaches $n$, suggesting that $\bv$ represents a solution to Problem~P5 and that its sum bandwidth requirement should now be checked; (2) $\sum_{m=1}^{d(\bv)}v_m<n$ and $d(\bv)=D$, in which case $\bv$ has reached the maximum depth but is not a solution to Problem~P5. These two types of nodes are marked with stripes and dots in Fig.~\ref{fig_trees_1}, respectively. If none of these conditions is satisfied, i.e., $d(\bv)<D$ and $\sum_{m=1}^{d(\bv)}<n$, then node $\bv$ has child nodes to be discovered because $\bv$ represents a partial solution that can possibly develop into full solutions. Conditioned on $\bv$, the sub-problem to find such solutions can be formulated. Specifically, by subtracting the accumulated number of selected tasks along the path to node $\bv$, the number of tasks that the sub-problem aims to select is 
\begin{equation}\label{eq: n_prime}
    n'=n-\sum_{m=1}^{d(\bv)}v_m,
\end{equation}
while the remaining latency budget is given by
\begin{equation}\label{eq: tau_prime}
    \tilde{\tau}_k'=\tilde{\tau}_k-\sum_{d=1}^{d(\bv)}f_d\left(n-\sum_{m=0}^{d-1} v_m\right),
\end{equation}
where we let $v_0=0$ for ease of notation. The sub-problem conditioned on $\bv$ that aims to find $\{\cS_{d(\bv)+1},\ldots,\cS_{D}\}$ with a total of $n'$ tasks is then mathematically written as
\begin{equation*}\text{(P4(}\bv\text{))}
\begin{aligned}
\mathrm{find}\quad\quad\quad  \cS'= & \bigcup_{m=d(\bv)+1}^D \cS_m   \\
   \mathrm{s.t. }\quad  \sum_{m=d(\bv)+1}^D & |\cS_m| = n', \\
                        \sum_{k\in \cS'} \quad & \rho_{\min, k} \leq 1-\rho(\bv), \\
                        \sum_{d=d(\bv)+1}^m & f_d(n_d)\leq \tilde{\tau}_k',\; \forall k \in \cS_m,\\ &  m=d(\bv)+1,\ldots,D,
\end{aligned}
\end{equation*}
where $\rho(\bv)$ denotes the accumulated sum bandwidth requirement along the path to node $\bv$. The sub-problem with large bandwidth, i.e. Problem P5($\bv$), can be similarly defined. Noting that $n_{d(\bv)+1}=n'$ is constant, the feasible subset of $\cK_{d(\bv)+1}$ is given by
\begin{equation}\label{eq: F_prime}
    \cF_{d(\bv)+1}=\{k\in\cK_{d(\bv)+1}|f_{d(\bv)+1}(n')\leq \tilde{\tau}_k'\}.
\end{equation}
Consequently, the set of child nodes of $\bv$, $\cN_\bv$, is discovered and given by
\begin{equation}
    \cN_\bv=\{[v_1,v_2,\ldots,v_{d(\bv)+1}]| v_{d(\bv)+1}=0,1,\ldots,M\},
\end{equation}
where $M = \min\{n',|\cF_{d(\bv)+1}|\}$.
\end{enumerate}

\begin{figure}[!t]
  \centering
  \includegraphics[width=85mm]{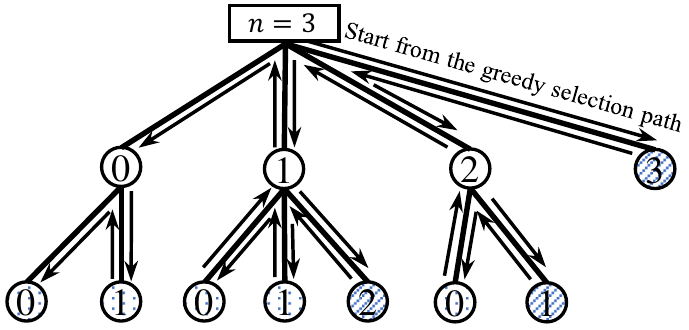}
  \caption{Example of the depth-first exhaustive tree-search when the maximum depth $D=2$ and $n=3$. The arrows indicate the depth-first search directions. }\label{fig_trees_1}
  \centering
\end{figure}

\subsubsection{Efficient Tree-search}
Consider a search over the tree constructed above to  solve Problem P4 for given $n$ scheduled users/tasks. In practice, when $n$ is large, it is impractical to store the entire tree in memory and  check  the feasibility of all nodes. The difficulty is overcome by designing an efficient tree-search algorithm. At the core of the design is to answer two questions: 1) \emph{In what sequence should the nodes be visited?}; 2) \emph{How can the results obtained in the preceding large-bandwidth case  be leveraged for complexity reduction?}

To answer the first question, recall that Algorithm~\ref{algorithm: greedy} is an efficient way to check the feasibility of Problem~P5, a relaxed version of the targeted Problem P4. We hence propose that the greedy user selection in Algorithm~\ref{algorithm: greedy} be first attempted because the infeasibility of Problem~P5 implies that of Problem~P4. Two design principles follow. First, among all child nodes of any node, the one with the maximum index $v_{d(\bv)}$, which is associated with greedily selecting tasks from the feasible set  $\cF_{d(\bv)}$,  should be prioritized. Second, the searching direction should prioritize depth rather than breadth, to ensure that we quickly reach the first leaf node following the greedy-selection path. The above principles motivate the use of \emph{depth-first search} (DFS) \cite{chen2011applied}. Specifically, upon visiting any node, among all its unvisited child nodes, the next step is to visit the one with the largest index, increasing the search depth  by $1$. Upon reaching a leaf node or a node without any unvisited child node, we \emph{backtrack} to visit its parent node and repeat the previous step. An example of  DFS is provided in   Fig.~\ref{fig_trees_1}. The stopping criterion is described as follows. During the search, if we encounter a node $\bv$ that represents a solution for  Problem~P5, such as the nodes with stripes in Fig.~\ref{fig_trees_1}, we check whether $\bv$ satisfies the bandwidth constraint. If so, then Problem~P4 is solved and the DFS stops. Otherwise, the DFS continues until all nodes have been visited. 

To answer the second question, there exist two methods for complexity reduction. First,  if the first leaf node discovered by DFS is not a solution for Problem~P5, we conclude that Problem~P4 is also infeasible, stopping the search. The second method is to prune nodes by evaluating if they satisfy the condition of being \emph{fathomed} defined as follows. 

\begin{definition}\label{def: fathom}
\emph{
A node $\bv$ is said to be \emph{fathomed} if deemed unable to develop into a solution of Problem~P5, namely that either of the following two conditions is satisfied. 
\begin{enumerate}
    \item[(1)] Node $\bv$ is a leaf node but not a solution to Problem~P5, i.e., $\sum_{m=1}^{d(\bv)}v_m<n$ and $d(\bv)=D$.
    \item[(2)] Node $\bv$ is not a leaf node, but its associated sub-problem, i.e., Problem~P5($\bv$), is infeasible.
\end{enumerate}}
\end{definition} 
Based on  the definition, fathomed nodes and their children  need not be further visited. If we can fathom nodes before visiting them, the associated   computational cost can be eliminated. Useful results for the purpose are presented in the  following theorem. 
\begin{theorem} [\emph{Fathoming Criteria}] \label{theorem: fathom}
\emph{\begin{enumerate}
    \item[(1)] If node $\bv$ with index $n^\dagger$ is fathomed, so are its sibling nodes with indices smaller than $n^\dagger$.
    \item[(2)] If all the child nodes of node $\bv$ are fathomed, so is  $\bv$;
\end{enumerate}}
\end{theorem}
\begin{proof}
 To prove (1), let node $\bw$ denote the parent node of node $\bv$. Since node $\bv$ is fathomed, Problem~P5($\bw$) has no solution such that $\vert\cS_{d(\bv)}\vert=n^\dagger$. By applying Lemma~\ref{lemma: local_optimality} to Problem~P5($\bw$), we can infer that there exists no solution such that $0\leq\vert\cS_{d(\bv)}\vert<n^\dagger$. Then for any child node of $\bw$, denoted by $\bv'$, with $v'_{d(\bv)}<v_{d(\bv)}$, Problem~P5($\bv'$) is infeasible. Therefore, $\bv'$ is fathomed. Furthermore,  (2) is obvious by the definition of fathomed nodes, completing the proof. 
\end{proof}

\begin{figure*}[!t]
  \centering
  \includegraphics[width=0.75\textwidth]{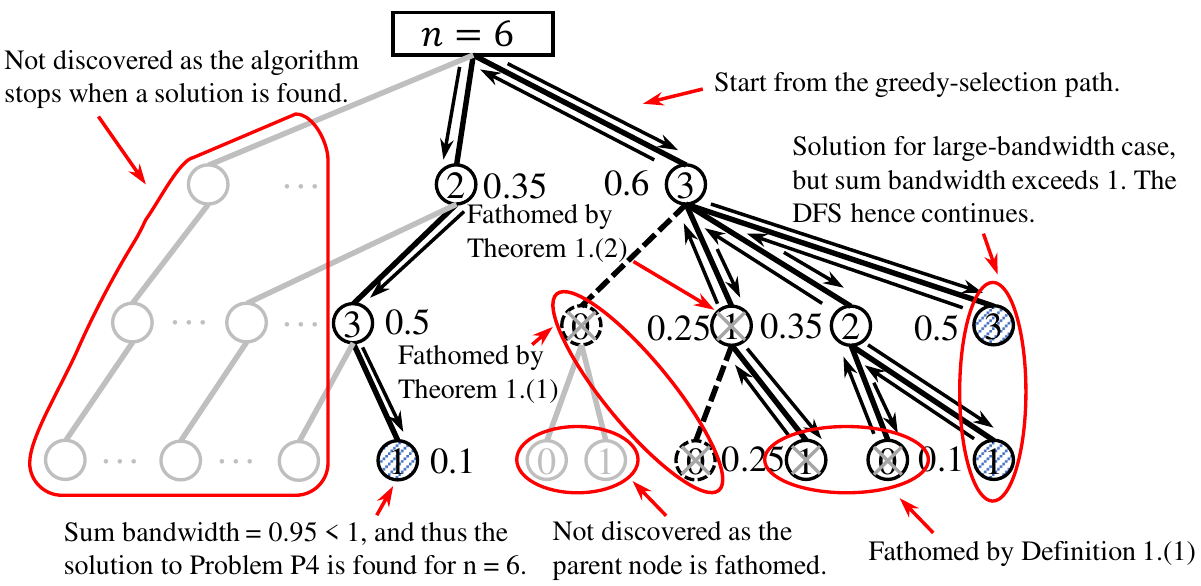}
  \caption{Searching for the solution to Problem~P4 for $n=6$ when the maximum depth $D=3$ using Algorithm~\ref{algorithm: optimal}. Nodes deemed unnecessary to be visited are plotted with dashed lines, and undiscovered nodes plotted with grey lines. }\label{fig_trees_2}
  \centering
\end{figure*}

Result (1) in Theorem~\ref{theorem: fathom} combined with DFS reduces the search complexity by pruning all unvisited sibling nodes of a fathomed node and their child nodes  without visiting them. On the other hand, result (2) allows 
the pruning of the parent node of a set of fathomed nodes, and thus the parent node's siblings, along with their children, can be pruned too. 

\begin{algorithm}[]
\caption{Optimal Algorithm for Solving Problem~P2}
\label{algorithm: optimal}
\textbf{Main Algorithm}\\
\textbf{Input:} The task index set $\cK=\{1,2,\ldots,K\}$, the bandwidth requirements $\{\rho_{\min, k}\}_{k=1}^{K}$, the latency constraints $\{\tilde{\tau}_k\}_{k=1}^{K}$;\\
\textbf{Initialize:} $\cK_d \triangleq \{k|k\in\cK, d_k=d\}$, $\cS=\emptyset$;\\
\textbf{for} $n=1,2,\cdots,K$ \textbf{do}\\
    \begin{enumerate}
        \item[]\!\!\!\!\!\!\!\!\! $\mathbf{v}_0 = [\;]$; \quad   \emph{-Initialize the root node as an empty vector.}\\
        \item[]\!\!\!\!\!\!\!\!\! Call DFS($\mathbf{v}_0$, $n$);
        \item[]\!\!\!\!\!\!\!\!\! \textbf{if} DFS($\mathbf{v}_0$, $n$) returns solution $\cS^{(n)}$ \textbf{then}\\
        \begin{enumerate}
            \item[]\!\!\!\!\!\!\!\!\! $\cS\leftarrow\cS^{(n)}$;\quad   \emph{-Admit $\cS^{(n)}$ as the new best solution.}
        \end{enumerate}
        \item[]\!\!\!\!\!\!\!\!\! \textbf{else} \textbf{break the loop}
    \end{enumerate}
\textbf{end for} \\
\textbf{return} $\cS$\\    
\vspace{0.5cm}
\textbf{function} DFS($\mathbf{v}$, $n$)\\
    \begin{enumerate}
        \item[] \!\!\!\!\!\!\!\!\! $d(\bv)\leftarrow$length($\bv$); \quad\quad\quad\quad\quad   \emph{-Get the depth of node $\bv$.}\\
        \item[] \!\!\!\!\!\!\!\!\! \textbf{Case 1:} $\sum_{m=1}^{d(\bv)} v_m = n$;
        \item[] \!\!\!\!\!\!\!\!\! \emph{(Node $\bv$ is a solution to the large-bandwidth case)}\\
        \begin{enumerate}
            \item[] \!\!\!\!\!\!\!\!\! Recover the task subset $\cS^{(n)}$ from $\bv$;
            \item[] \!\!\!\!\!\!\!\!\! \textbf{if} $\sum_{k\in\cS^{(n)}}\rho_{\min, k}\leq 1$ \quad \emph{-Evaluate the sum bandwidth requirement.}
            \begin{enumerate}
                \item[] \!\!\!\!\!\!\!\!\! \textbf{return} $\cS^{(n)}$ as the solution to Problem~P4
            \end{enumerate}
            \item[]  \!\!\!\!\!\!\!\!\! \textbf{else}
            \begin{enumerate}
                \item[] \!\!\!\!\!\!\!\!\! \textbf{return} no solution
            \end{enumerate}
        \end{enumerate}
        \item[] \!\!\!\!\!\!\!\!\! \textbf{Case 2:} $\sum_{m=1}^{d(\bv)} v_m < n$ and $d(\bv)=D$;
        \item[] \!\!\!\!\!\!\!\!\! \emph{(Node $\bv$ is a leaf node but not a solution to the large-bandwidth case)}\\
        \begin{enumerate}
            \item[] \!\!\!\!\!\!\!\!\! Mark node $\bv$ as fathomed.
            \item[] \!\!\!\!\!\!\!\!\! \textbf{return} no solution 
        \end{enumerate}
        \item[] \!\!\!\!\!\!\!\!\! \textbf{Case 3:} $\sum_{m=1}^{d(\bv)} v_m < n$ and $d(\bv)<D$;
        \item[] \!\!\!\!\!\!\!\!\! \emph{(Node $\bv$ is not a solution but has child nodes that can possibly develop into solutions. )}\\
        \begin{enumerate}
            \item[] \!\!\!\!\!\!\!\!\! Find $n'$ and $\tilde{\tau}'_k$ by \eqref{eq: n_prime} and \eqref{eq: tau_prime};
            \item[] \!\!\!\!\!\!\!\!\! Find the feasible subset $\cF_{d(\bv)+1}$ by \eqref{eq: F_prime};
            \item[] \!\!\!\!\!\!\!\!\! $M\leftarrow\min\{n',|\cF_{d(\bv)+1}|\}$;
            \item[] \!\!\!\!\!\!\!\!\! \textbf{for} $v_{d(\bv)+1}=M,M-1,\ldots,0$ \quad\quad \emph{-Prioritizing child nodes with larger index.}\\ 
            \begin{enumerate}
                \item[] \!\!\!\!\!\!\!\!\! $\bv' \leftarrow [\bv, v_{d(\bv)+1}]$;
                \item[] \!\!\!\!\!\!\!\!\! Call DFS($\mathbf{v}'$, $n$);
                \item[] \!\!\!\!\!\!\!\!\! \textbf{if} DFS($\mathbf{v}'$, $n$) returns solution $\cS^{(n)}$ \textbf{then}
                \begin{enumerate}
                    \item[] \!\!\!\!\!\!\!\!\! \textbf{return} $\cS^{(n)}$ as the solution to Problem~P4
                \end{enumerate}
                \item[]  \!\!\!\!\!\!\!\!\! \textbf{else} \textbf{if} $\bv'$ is fathomed \textbf{then}
                \begin{enumerate}
                    \item[] \!\!\!\!\!\!\!\!\! Mark all unvisited child nodes of $\bv$ as fathomed. 
                    \item[] \!\!\!\!\!\!\!\!\! \textbf{break the loop}
                \end{enumerate}
            \end{enumerate}
            \item[] \!\!\!\!\!\!\!\!\! \textbf{if} all child nodes of $\bv$ is fathomed \textbf{then}
                \begin{enumerate}
                    \item[] \!\!\!\!\!\!\!\!\! Mark $\bv$ as fathomed. 
                \end{enumerate}
            \item[] \!\!\!\!\!\!\!\!\! \textbf{return} no solution
        \end{enumerate}
    \end{enumerate}
\textbf{end function}
\end{algorithm}

\begin{example}
\emph{To illustrate the said complexity reduction,  Fig.~\ref{fig_trees_2} shows the complete process of tree-searching for Problem~P4 with $n=6$. The de facto complexity reduction comes from the nodes deemed unnecessary to be visited, which are plotted using dashed lines, and nodes not discovered, plotted using gray lines. It can be observed that as the DFS fathoms the leaf node $[3, 1, 1]$ by Definition~\ref{def: fathom}, $[3, 1, 0]$ is fathomed by Theorem~\ref{theorem: fathom}.(2), and then the parent node $[3, 1]$ is fathomed by Theorem~\ref{theorem: fathom}.(1) since all its child nodes are fathomed. Applying Theorem~\ref{theorem: fathom}.(2) on $[3, 1]$, $[3, 0]$ can be immediately fathomed and skipped, of which all the child nodes need not be discovered. The DFS then explores a new branch and subsequently finds a solution $[2, 3, 1]$.}
\end{example}

It is worth pointing out by Theorem~\ref{theorem: fathom}, if the first leaf node that DFS encounters is fathomed, then we can infer that all other nodes are fathomed in a chain. Hence, the aforementioned first method for complexity reduction is in fact a special case of the latter one.

By now, the algorithm for efficient tree-search is in order, which solves  Problem~P4. Recall that  the optimal policy  for resource allocation with early exiting results from solving Problem~P2. Then the problem  can be solved by sequentially solving Problem~P4 via the tree-search  for $n=1,\ldots,K$ until we find the smallest $n^\star$ such that $(n^\star+1)$ makes Problem~P4 infeasible. The proof of optimality is similar to that of Proposition~\ref{prob: sub_prob_infeas} and thus omitted for brevity. The detailed algorithm is described  in Algorithm~\ref{algorithm: optimal}.

\subsubsection{Complexity Analysis}
The complexity analysis for Algorithm~\ref{algorithm: optimal} is given as below. Node $\bv$ has at most $\cK_{d(\bv)+1}$ child nodes, and thus the entire search tree for Problem~P4 has at most $\sum_{d=1}^D 
\prod_{m=1}^d|\cK_m|$ nodes. Note that $\sum_{d=1}^D|\cK_d|=K$, and hence the number of nodes is upper bounded by $\sum_{d=1}^D(\frac{K}{d})^d$. The computation complexity of visiting any node is $\cO(\max\{K,D\})$. Hence, the complexity of solving Problem~P4 is $\cO(\max\{K,D\}(\frac{K}{D})^D)$, which is repeated for $n=1,\ldots,K$ times in the worst case. Therefore, the complexity of Algorithm~\ref{algorithm: optimal} is given as $\cO(\max\{K,D\}K(\frac{K}{D})^D)$. Due to the NP-completeness of Problem P2, it is anticipated that the complexity is exponential in $D$. However, in practice, $D$ is usually much smaller than $K$ and unlikely to scale up. Therefore, the complexity is given as $\cO(K^{D+2})$, which is polynomial in the number of tasks $K$, even without considering the complexity reduction by pruning.

\section{Experimental Results}
\label{sec: experiments}

\subsection{Experimental Settings} The experimental settings are as follows unless specified otherwise. Following a common approach in the literature, the  arrivals of multiuser task requests at the server are modeled using a Poisson process with the arrival rate $\lambda$ varying from $10$ to more than $100$ tasks per second\cite{LazyBatch,AliSC20}. Each request  is tagged using a random mark specifying the associated user's channel gain and QoS requirements, which are distributed as follows. The marks are independent and identically distributed (i.i.d.) with  
the  channel gain following  Rayleigh fading with average path loss set as $10^{-3}$, the latency requirement $\tau_k$ uniformly distributed between $0.5$ and $2$ seconds, and the resultant  exit point $d_k$ uniformly selected from the set $\{1,2,3\}$. The task requests within a single epoch are assumed to come from different users. The total bandwidth is set as $B = 20$ {MHz}. The transmit  \emph{signal-to-noise ratio} (SNR) is uniform for users and set as $p_k/N_0=20$ {dB}. The size of all feature vectors is fixed as $\ell_k=10$ {KB}. {Following \cite{pmlr-v70-bolukbasi17a}, the inference model is ResNet-50 trained on ImageNet with three candidate exit points, each leading to a classifier.} Specifically, two intermediate classifiers are placed after ``res3a'' and ``res4a'' layers respectively and the final classifier follows the last layer (see Fig.~\ref{fig_diagram_earlyexit}). Therefore, the maximum number of traversed blocks is $D=3$.  The edge server is equipped with a popular edge GPU, NVIDIA JETSON TX2 GPU, of which the computation latency profiles, i.e., $f(n_\mathsf{b})$ and $f_d(n_d)$, follow the measurements  reported in \cite{nvidia2018,pmlr-v70-bolukbasi17a}. The duration of communication and computation slots is set as $T=250$ {ms}. Performance comparison is based on  the metric of task completion rate, defined as the ratio of completed tasks to all generated tasks. 

{Two benchmarking schemes from standard protocols of batched inference are adopted (see e.g.,~\cite{nvidia2018,AliSC20}). One is single-instance inference while the other is featured with the  optimized batch size, as described below.}
\begin{itemize}
    \item \emph{Single-instance inference}: The arrived tasks requests wait in a \emph{first-in-first-out} (FIFO) queue. When idle, the server accepts  the first task from the queue head, and allocates all bandwidth to this task for feature uploading. Upon receiving the feature vector, the server executes single-instance inference without batching. A task is dropped if the waiting time exceeds the allowed  latency. 
    \item \emph{Static batching}: For this scheme, the server has a pre-determined batch size and a pre-determined timeout parameter. Upon the arrival of a task request, the server schedules feature uploading and stores the received feature vector in a buffer. The server assembles all tasks into a batch for parallel inference either the number of buffered tasks reaches that fixed batch size, or the elapsed waiting time exceeds the timeout preset. Due to the prohibitive complexity of parametric optimization by grid search, the batch size and timeout parameter optimized for  $\lambda=50$ tasks per second are used for all arrival rates.
\end{itemize}

\subsection{Resource Allocation without Early Exiting}

\begin{figure}[t]
  \centering
  \includegraphics[width=80mm]{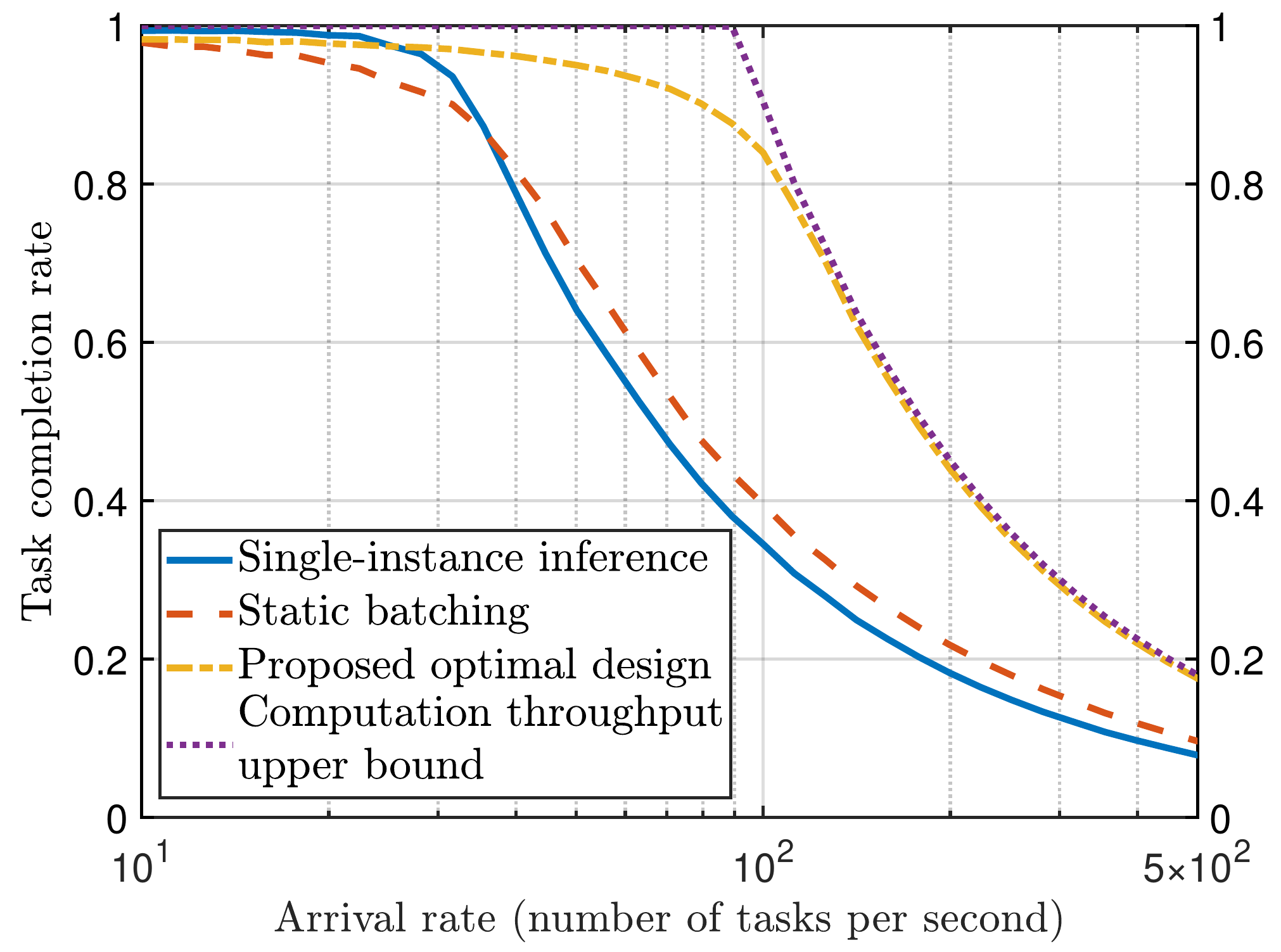}
  \caption{Completion rate versus task arrival rate in full-network inference without early exiting. }\label{fig_non_ee_1}
  \centering
\end{figure}

\begin{figure}[t]
  \centering
  \includegraphics[width=80mm]{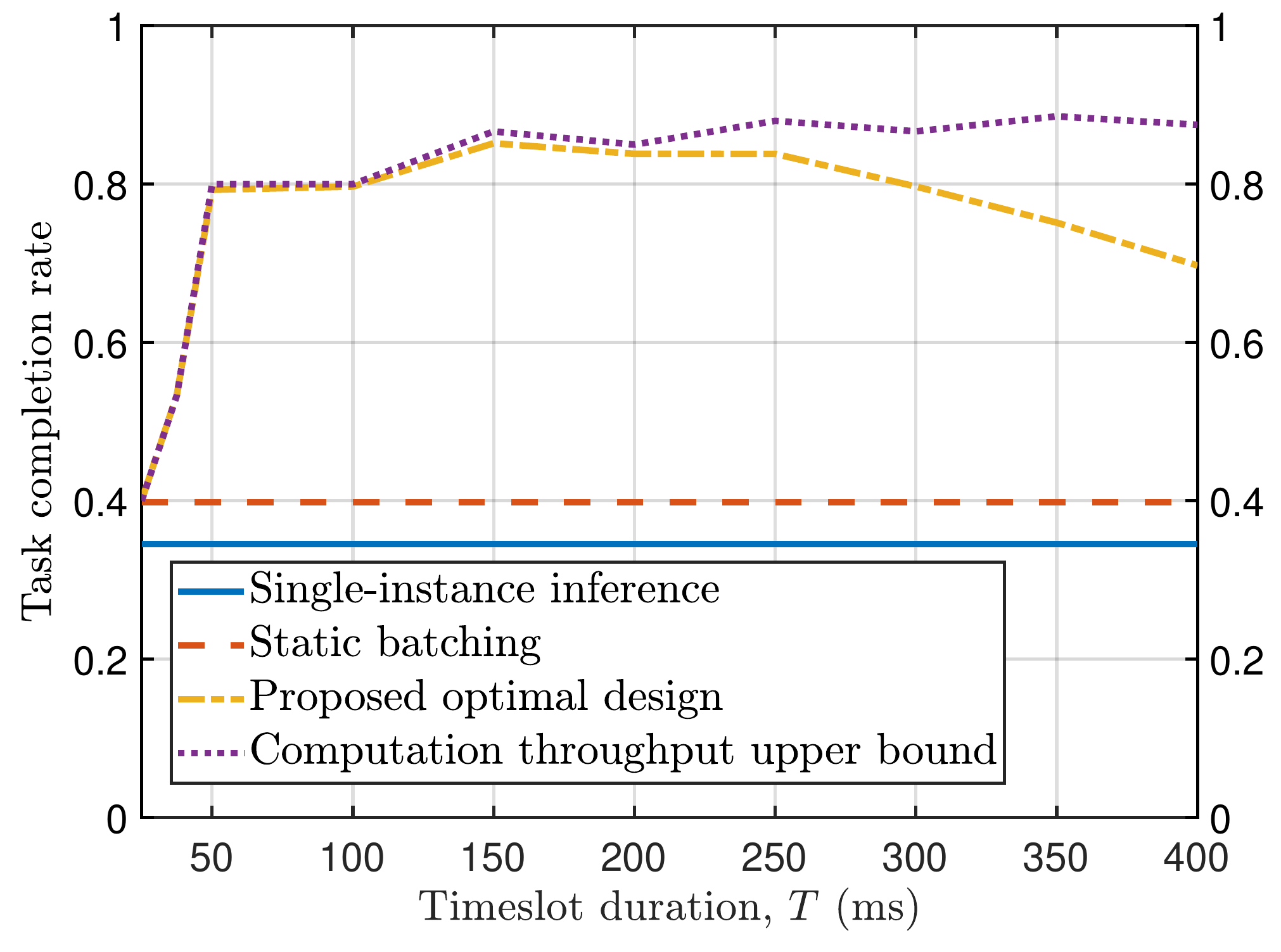}
  \caption{Completion rate versus the time-slot duration, $T$, in full-network inference without early exiting. }
 
  \label{fig_non_ee_timeslot}
  \centering
  
\end{figure}

\begin{figure}[t]
  \centering
  \includegraphics[width=80mm]{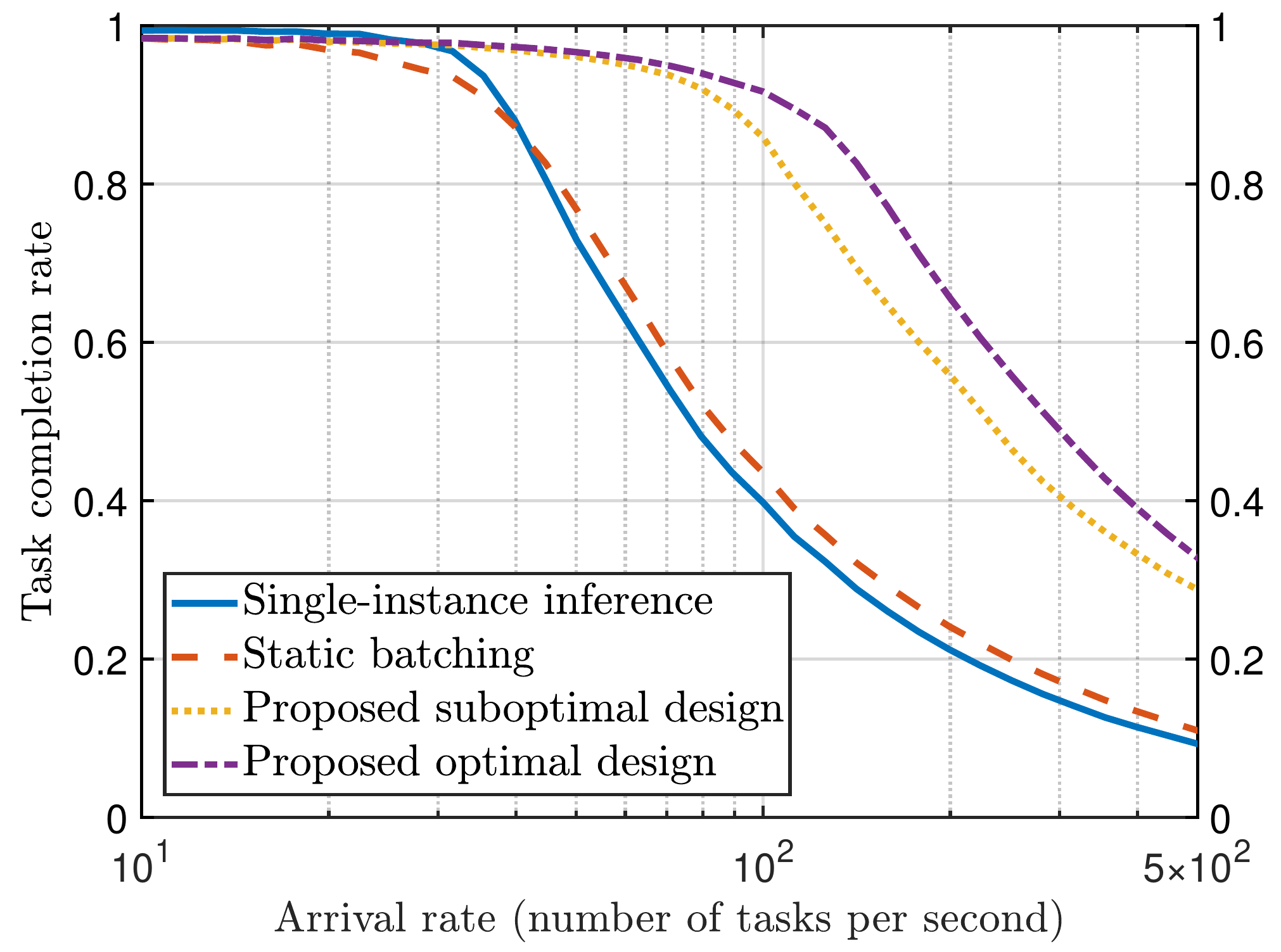}
  \caption{Completion rate versus task arrival rate in the case with early exiting. }\label{fig_ee_1}
  \centering
\end{figure}

In Fig.~\ref{fig_non_ee_1}, the curves of task completion rate versus the task arrival rate are plotted, {where ``proposed optimal design'' refers to Algorithm~\ref{algorithm: opt_p1}.} One can observe that the proposed design can attain large performance gain by exploiting joint communication-and-computation resource allocation for more efficient edge inference. For instance, at  the arrival rate of $100$ task/s, the proposed algorithm  achieves the completion rate of $83.5\%$, a more than $100\%$ improvement over  benchmarking  schemes. However, single-instance inference performs slightly better than the proposed algorithm in the region of sparse task arrival. This results from the proposed time-slotted design incurring waiting time even when the server is idle. {The edge server's maximum achievable throughput is no larger than that in the case with  an arbitrarily large batch size, therefore imposing an upper bound on completion rate (i.e., normalized throughput) as plotted in Fig.~\ref{fig_non_ee_1} with the legend ``computation throughput upper bound''. A key  observation is that the task arrival rate begins to exceed the maximum achievable throughput at around $90$ tasks per second. Before this critical point, the completion rate loss is attributed to the communication bottleneck and queuing timeouts. After this  point, the completion rate approaches its  upper bound while both exhibit a downward trend as the arrival rate grows.}

Fig.~\ref{fig_non_ee_timeslot} shows the completion rate versus the time-slot duration, $T$, for  communication/computation. The arrival rate is set as $\lambda = 100$ task/s. The computation throughput upper bound is derived with a slight difference compared with that in Fig.~\ref{fig_non_ee_timeslot}, as it now reflects the maximum batch size limited by the timeslot duration. One can observe that the slot duration being too short  limits the maximum batch size and thereby introduces a throughput bottleneck to the proposed algorithm. As the slot duration increases and the bottleneck is alleviated, its completion rate can be observed to first grow,  subsequently be saturated, and finally reduce to deviate from its  upper bound. The saturation region corresponds to the case where the computation resource is under-utilized and the bandwidth constraint becomes the main performance-limiting factor. The degradation in the last region is due to the excessive latency caused by long slot duration. 

\subsection{Resource Allocation with Early Exiting}

\begin{figure}[t]
  \centering
  \includegraphics[width=80mm]{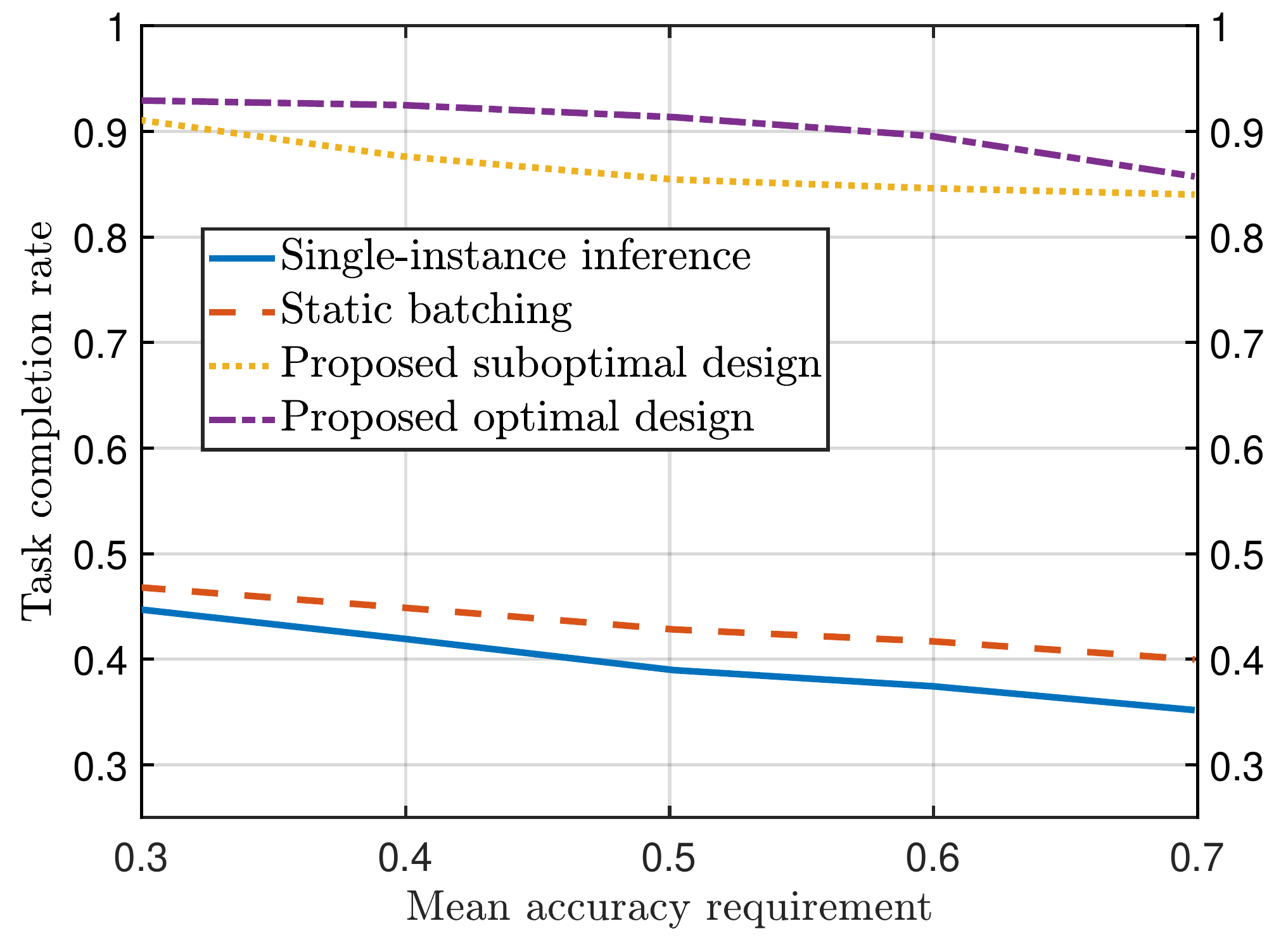}
  \caption{Completion rate versus the mean accuracy requirement of inference requests. }\label{fig_comp_rate_vs_acc}
  \centering
\end{figure}

Fig.~\ref{fig_ee_1} depicts the curves of task completion rate versus task arrival rate in the early exiting case. First, comparing  Figs.~\ref{fig_ee_1} and \ref{fig_non_ee_1}, one can observe that the incorporation of early exits leads to substantial improvements in the completion rate as a result of more efficient utilization of computation resources. As before, the proposed algorithms substantially outperform the  baseline schemes when  the arrival rate is sufficiently large.  For example, at the arrival rate of $100$ task/s, the proposed  sub-optimal and optimal algorithms   improve the completion rate over benchmarking schemes by about  $100\%$. On the other hand, the performance of the sub-optimal algorithm approaches that of the optimal one  at low arrival rates when the computation resources are abundant. As the arrival rate increases, the performance gap between them emerges but is observed to be  approximately constant, which is contributed by the now non-negligible effect of shrinking batch size.

{To study the effects of accuracy requirements on the task completion rate, we vary the mean accuracy requirement of inference requests. The results are presented  in Fig.~\ref{fig_comp_rate_vs_acc}, showing a decrease in completion rates as the mean of required  accuracy  increases. This is intuitive as to achieve   higher accuracy asks for  more computation resources or equivalently  traversing  deeper layers of the neural network. When  the mean accuracy  approaches its maximum achievable by  full-network inference (i.e., an accuracy of 74.9\% on ImageNet), almost all tasks require executing the entire network. On the other hand, when the required mean accuracy is low, it is the limited bandwidth rather than computation power that causes   the throughput bottleneck. In both two cases, the sub-optimality caused by ignoring the shrinking batch size becomes negligible, which explains the small performance gaps between the proposed suboptimal and optimal design towards  two end points of the axis of mean accuracy requirement. Next, we vary the mean latency requirement to study its impact on the completion rate. The latency requirements are uniformly distributed on the interval $[\bar{\tau}-250,\bar{\tau}+250]$ (ms) where $\bar{\tau}$ is the mean latency varying between $250$ ms and $1750$ ms. The control parameters of all proposed and baseline approaches are re-optimized as the latency distribution changes. We can see from Fig.~\ref{fig_comp_rate_vs_lat} that the proposed approaches and static batching benefit from loosening the latency constraints, mainly because larger batch sizes lead to higher    inference throughput. On the contrary, this does not yield any performance gain in the case of single-instance inference since its maximum throughput has already been achieved. Combining Fig.~\ref{fig_comp_rate_vs_acc} and Fig.~\ref{fig_comp_rate_vs_lat} allows  the tradeoff between accuracy and latency requirements,  which is regulated by the  completion rate,  to be quantified. }

\begin{figure}[t]
  \centering
  \includegraphics[width=80mm]{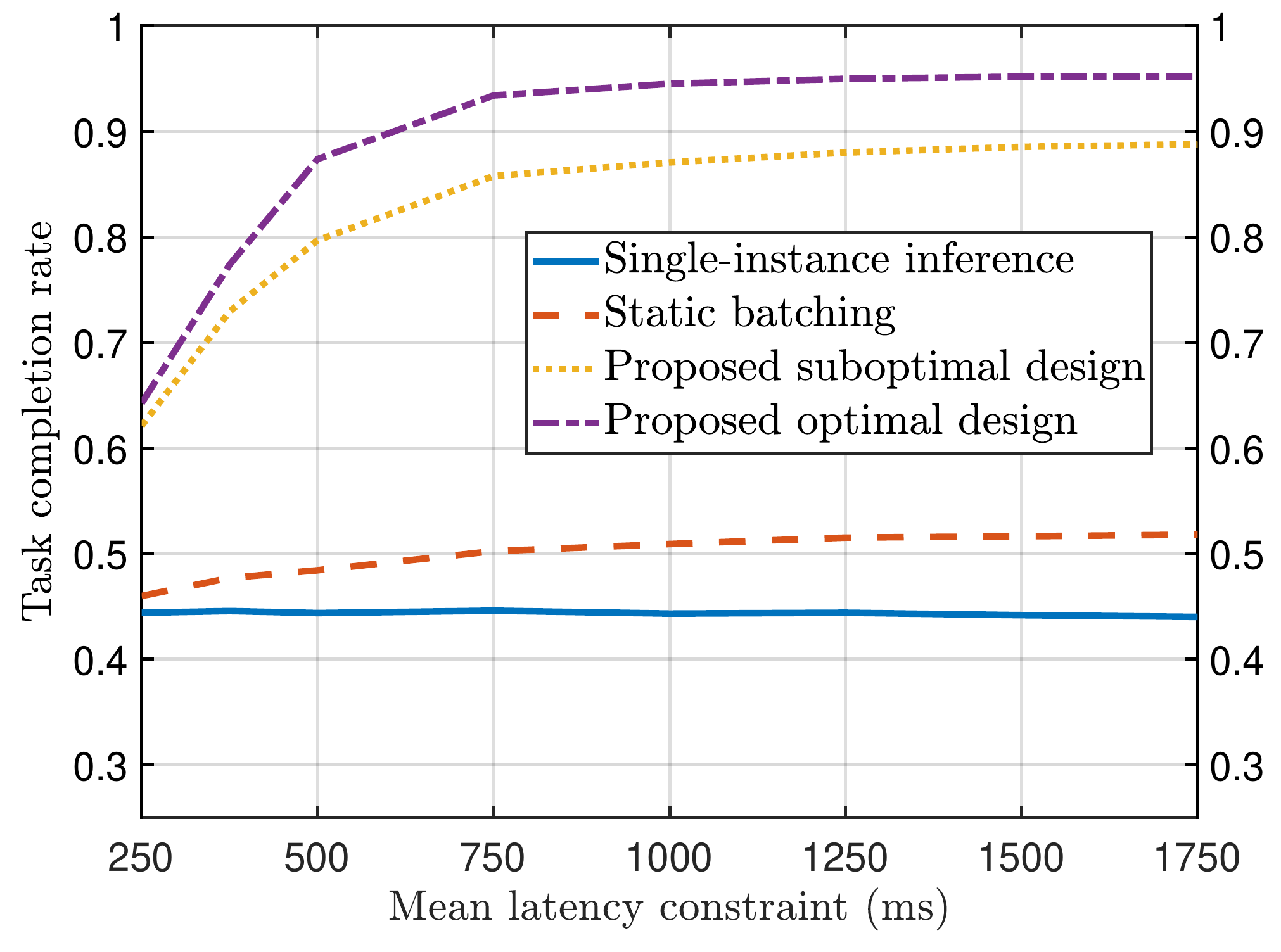}
  \caption{Completion rate versus the mean latency requirement of inference requests. }\label{fig_comp_rate_vs_lat}
  \centering
\end{figure}

\begin{figure}[t]
  \centering
  \includegraphics[width=80mm]{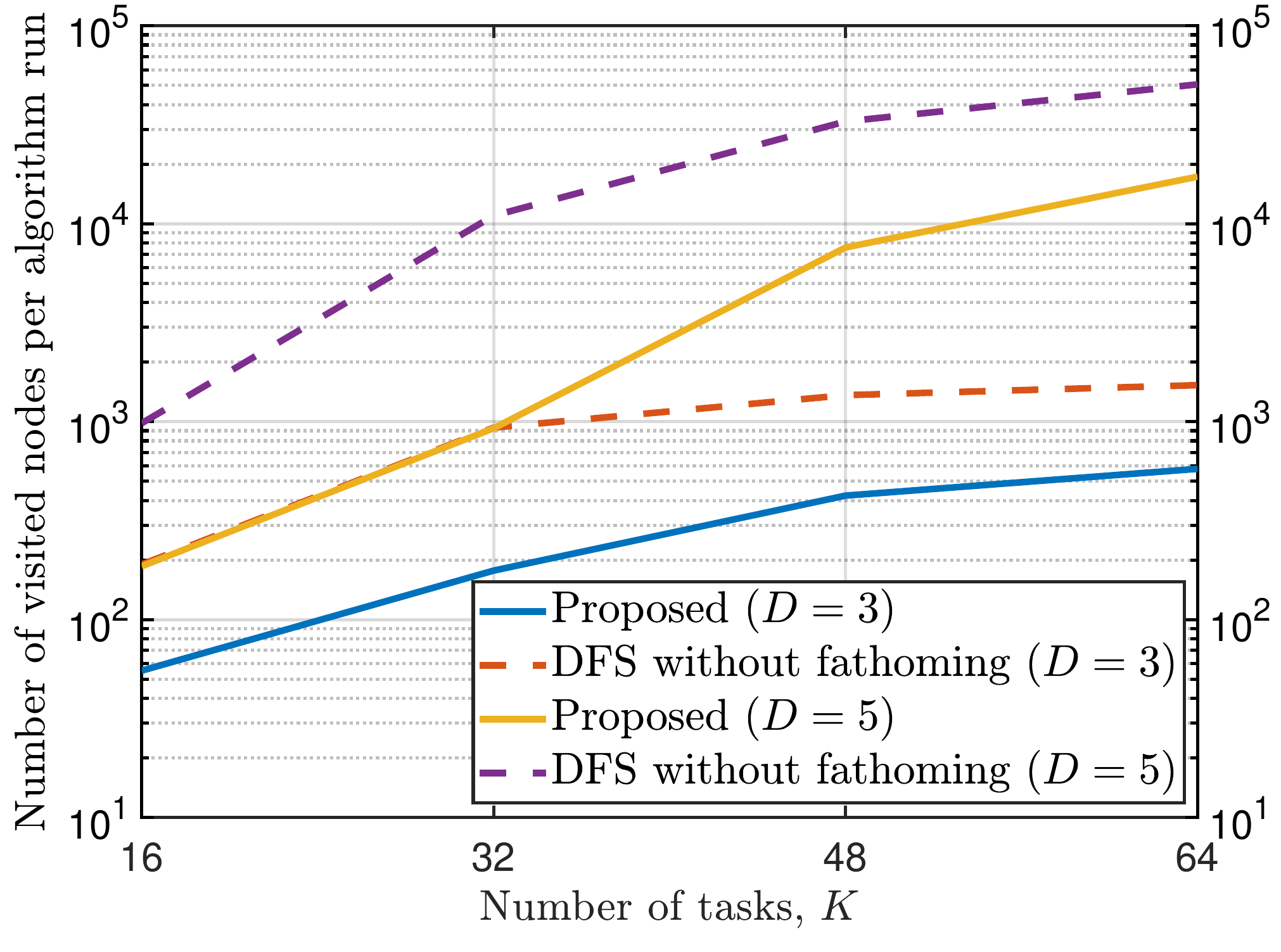}
  \caption{Number of visited nodes per algorithmic run versus the number of tasks to schedule per epoch, $K$. }\label{fig_complexity}
  \centering
\end{figure}
Last, we study the computational complexity of the proposed Algorithm~\ref{algorithm: optimal} for efficient tree-search, which finds the  optimal resource allocation policy with early exiting, as opposed to the benchmarking  design without exploiting node fathoming and pruning. The number of layer blocks in the inference model is $D = \{3, 5\}$. 
In Fig.~\ref{fig_complexity}, the number of visited nodes per algorithmic  run is plotted against the number of tasks to schedule per epoch. It can be seen that the proposed technique of node fathoming dramatically reduces the complexity of tree-search, for example, by about $12$-time for the case of $D =5$ and $K=32$.

\section{Conclusions}
\label{sec: conclusions}
In this paper, we have studied  the joint communication-and-computation resource allocation for multiuser edge inference featuring  batching and early exiting under the criterion of throughput maximization.  In the case without early exiting, it has been found that the optimal policy possesses a threshold-based structure such that among tasks with latency budget above  the  threshold, those with the best channels are prioritized in receiving the inference service.  
In the early exiting case, we advocate finding the optimal policy, which is a combinatorial  problem,  via an efficient tree-search algorithm. To this end, we analyze the special case where the bandwidth is sufficient, and reveal that the optimal policy greedily schedules those users with relatively loose latency and accuracy   requirements and executes their  tasks. These findings are leveraged to design the technique of node fathoming to substantially reduce the tree-search complexity. 

This work opens up the direction of resource allocation for efficient multiuser edge inference with batching and early exiting. Several issues warrant follow-up  investigation including adaptive split inference to support flexible computation offloading,  the generalization to a multi-server system requiring local balancing across servers and the case of heterogeneous task types requiring different models, and the incorporation of more sophisticated physical layer techniques such as \emph{multiple-input-multiple-output} (MIMO) and \emph{non-orthogonal multi-access} (NOMA). 

\appendix
\subsection{Proof of Lemma~\ref{lemma: local_optimality}}
Assume that no solution exists such that $|\cS_1|=n^\dagger\leq|\cF_1|$ but there exists a solution $\{\cS_1, \cS_2, \ldots, \cS_D\}$ to Problem P5 such that 
$0\leq|\cS_1|<n^\dagger$, and let $\{n_d\}$ denote the consequent block-wise batch size. Hence, $\cF_1\setminus\cS_1$ is non-empty. Let $r$ be an arbitrary element of $\cF_1\setminus\cS_1$. Note that $|\cS_1|<n^\dagger\leq|\cF_1|<n$, then there exist at least one non-empty set in $\{\cS_2, \cS_3, \ldots, \cS_D\}$, say $\cS_m\neq\emptyset$. Let $s$ be an arbitrary element of $\cS_m$. Now, consider $\{\cS_1\cup\{r\}, \cS_2, \ldots,\cS_m\setminus\{s\},\ldots, \cS_D\}$ as a new candidate point, and let $\{n'_d\}$ denote its consequent block-wise batch size. From (\ref{equation: n_d}), we know that $n'_d\leq n_d$ holds for $d=1,\ldots,D$ and thus $f(n'_d)\leq f(n_d)$ as $f(n_d)$ is a non-increasing function. It is then easy to verify that $\{\cS_1\cup\{r\}, \cS_2, \ldots,\cS_m\setminus\{s\},\ldots, \cS_D\}$ is also a solution to Problem P5. This process can be repeated until we find a solution such that $|\cS_1|=n^\dagger$, which leads to a contradiction. This completes the proof of Lemma~\ref{lemma: local_optimality}.
\bibliographystyle{IEEEtran}
\bibliography{Edge-Inference}

% Generated by IEEEtran.bst, version: 1.14 (2015/08/26)
\begin{thebibliography}{10}
\providecommand{\url}[1]{#1}
\csname url@samestyle\endcsname
\providecommand{\newblock}{\relax}
\providecommand{\bibinfo}[2]{#2}
\providecommand{\BIBentrySTDinterwordspacing}{\spaceskip=0pt\relax}
\providecommand{\BIBentryALTinterwordstretchfactor}{4}
\providecommand{\BIBentryALTinterwordspacing}{\spaceskip=\fontdimen2\font plus
\BIBentryALTinterwordstretchfactor\fontdimen3\font minus
  \fontdimen4\font\relax}
\providecommand{\BIBforeignlanguage}[2]{{%
\expandafter\ifx\csname l@#1\endcsname\relax
\typeout{** WARNING: IEEEtran.bst: No hyphenation pattern has been}%
\typeout{** loaded for the language `#1'. Using the pattern for}%
\typeout{** the default language instead.}%
\else
\language=\csname l@#1\endcsname
\fi
#2}}
\providecommand{\BIBdecl}{\relax}
\BIBdecl

\bibitem{Letaief2022JSAC}
K.~B. Letaief, Y.~Shi, J.~Lu, and J.~Lu, ``Edge artificial intelligence for
  6{G}: Vision, enabling technologies, and applications,'' \emph{IEEE J. Sel.
  Areas Commun.}, vol.~40, no.~1, pp. 5--36, 2022.

\bibitem{Zhang2020CM}
J.~Shao and J.~Zhang, ``Communication-computation trade-off in
  resource-constrained edge inference,'' \emph{IEEE Commun. Mag.}, vol.~58,
  no.~12, pp. 20--26, 2020.

\bibitem{Niu2019Infocom}
W.~Shi, Y.~Hou, S.~Zhou, Z.~Niu, Y.~Zhang, and L.~Geng, ``Improving device-edge
  cooperative inference of deep learning via 2-step pruning,'' in \emph{Proc.
  IEEE Conf. Comput. Commun. Workshops (INFOCOM WKSHPS)}, Paris, France, Apr
  29--May 2, 2019.

\bibitem{Deniz2020SPAWC}
M.~Jankowski, D.~Gündüz, and K.~Mikolajczyk, ``Joint device-edge inference
  over wireless links with pruning,'' in \emph{Proc. IEEE Int. Workshop Signal
  Process. Adv. Wireless Commun. (SPAWC)}, Atlanta, GA, USA, May 26--29, 2020.

\bibitem{LQ2021arxiv}
Q.~Lan, Q.~Zeng, P.~Popovski, D.~G\"und\"uz, and K.~Huang, ``Progressive
  feature transmission for split classification at the wireless edge,''
  \emph{\emph{to appear in} IEEE Trans. Wireless Commun.}, 2022.

\bibitem{Shao2021arxiv}
J.~Shao, Y.~Mao, and J.~Zhang, ``Task-oriented communication for multi-device
  cooperative edge inference,'' \emph{\emph{to appear in} IEEE Trans. Wireless
  Commun.}, 2022.

\bibitem{Zhou2020IoTJ}
X.~Huang and S.~Zhou, ``Dynamic compression ratio selection for edge inference
  systems with hard deadlines,'' \emph{IEEE Internet Things J.}, vol.~7, no.~9,
  pp. 8800--8810, 2020.

\bibitem{Chen2021IOTJ}
X.~Tang, X.~Chen, L.~Zeng, S.~Yu, and L.~Chen, ``Joint multiuser {D}{N}{N}
  partitioning and computational resource allocation for collaborative edge
  intelligence,'' \emph{IEEE Internet Things J.}, vol.~8, no.~12, pp.
  9511--9522, 2021.

\bibitem{Chen2020TWC}
E.~Li, L.~Zeng, Z.~Zhou, and X.~Chen, ``Edge {A}{I}: On-demand accelerating
  deep neural network inference via edge computing,'' \emph{IEEE Trans.
  Wireless Commun.}, vol.~19, no.~1, pp. 447--457, 2020.

\bibitem{dualTSP}
Y.~Wang, J.~Shen, T.-K. Hu, P.~Xu, T.~Nguyen, R.~Baraniuk, Z.~Wang, and Y.~Lin,
  ``Dual dynamic inference: Enabling more efficient, adaptive, and controllable
  deep inference,'' \emph{IEEE J. Sel. Top. Signal Process.}, vol.~14, no.~4,
  pp. 623--633, 2020.

\bibitem{nvidia2018}
NVIDIA, ``{N}{V}{I}{D}{I}{A} {A}{I} inference platform technical overview,''
  [Online]
  \url{https://www.nvidia.com/en-us/data-center/resources/inference-technical-overview/},
  2018.

\bibitem{TensorRT}
N{V}{I}DIA, ``{N}{V}{I}{D}{I}{A} {T}ensor{R}{T} documentation,'' [Online]
  \url{https://docs.nvidia.com/deeplearning/tensorrt/developer-guide/index.html},
  2022.

\bibitem{pmlr-v70-bolukbasi17a}
T.~Bolukbasi, J.~Wang, O.~Dekel, and V.~Saligrama, ``Adaptive neural networks
  for efficient inference,'' in \emph{Proc. Int. Conf. Mach. Learn. (ICML)},
  Sydney, Australia, Aug 6--11 2017.

\bibitem{MSDNet}
G.~Huang, D.~Chen, T.~Li, F.~Wu, L.~van~der Maaten, and K.~Weinberger,
  ``Multi-scale dense networks for resource efficient image classification,''
  in \emph{Proc. Int. Conf. Learn. Represent. (ICLR)}, Vancouver, BC, Canada,
  Apr 30--May 3, 2018.

\bibitem{GX2020CM}
G.~Zhu, D.~Liu, Y.~Du, C.~You, J.~Zhang, and K.~Huang, ``Toward an intelligent
  edge: Wireless communication meets machine learning,'' \emph{IEEE Commun.
  Mag.}, vol.~58, no.~1, pp. 19--25, 2020.

\bibitem{AliSC20}
A.~Ali, R.~Pinciroli, F.~Yan, and E.~Smirni, ``Batch: Machine learning
  inference serving on serverless platforms with adaptive batching,'' in
  \emph{Proc. Int. Conf. High Perform. Comput. Netw. Storage Anal.}, Nov 9--19,
  2020.

\bibitem{LazyBatch}
Y.~Choi, Y.~Kim, and M.~Rhu, ``Lazy batching: An {S}{L}{A}-aware batching
  system for cloud machine learning inference,'' in \emph{Proc. IEEE Int. Symp.
  High Perform. Comput. Archit. (HPCA)}, Feb 27--Mar 3, 2021.

\bibitem{E2bird}
W.~Cui, Q.~Chen, H.~Zhao, M.~Wei, X.~Tang, and M.~Guo, ``E2bird: Enhanced
  elastic batch for improving responsiveness and throughput of deep learning
  services,'' \emph{IEEE Trans. Parallel Distrib. Syst.}, vol.~32, no.~6, pp.
  1307--1321, 2021.

\bibitem{Inoue2021}
Y.~Inoue, ``Queueing analysis of {G}{P}{U}-based inference servers with dynamic
  batching: A closed-form characterization,'' \emph{Perform. Eval.}, vol. 147,
  2021.

\bibitem{goodfellow2016deep}
I.~Goodfellow, Y.~Bengio, A.~Courville, and Y.~Bengio, \emph{Deep
  Learning}.\hskip 1em plus 0.5em minus 0.4em\relax Cambrige, MA, USA: MIT
  Press, 2016.

\bibitem{KORTE1981}
B.~Korte and R.~Schrader, ``On the existence of fast approximation schemes,''
  in \emph{Nonlinear Programming}.\hskip 1em plus 0.5em minus 0.4em\relax
  Academic Press, 1981, vol.~4, pp. 415--437.

\bibitem{chen2011applied}
D.-S. Chen, R.~G. Batson, and Y.~Dang, \emph{Applied Integer Programming:
  Modeling and Solution}.\hskip 1em plus 0.5em minus 0.4em\relax Hoboken, NJ,
  USA: John Wiley \& Sons, 2011.

\end{thebibliography}

\end{document}